\documentclass[journal,12pt,onecolumn]{IEEEtran}

\usepackage{amsmath,amsfonts}
\usepackage{amssymb,verbatim,amsfonts,amsmath,amsthm}
\usepackage{geometry}
\usepackage{cite}
\usepackage{hyperref}
\usepackage{bm}

\usepackage{longtable}
\usepackage{booktabs}
\usepackage[justification=centering]{caption}
\usepackage{graphicx}
\usepackage[figuresright]{rotating}

\usepackage{makecell}

\newtheorem{theorem}{Theorem}
\newtheorem{lemma}[theorem]{Lemma}

\newtheorem{proposition}[theorem]{Proposition}
\newtheorem{example}[theorem]{Example}
\newtheorem{remark}[theorem]{Remark}
\newtheorem{definition}[theorem]{Definition}

\begin{document}

\title{Optimal and Almost Optimal Locally Repairable Codes from Hyperelliptic Curves}
\author{Junjie Huang\thanks{J. Huang is with the School of Mathematics, Sun Yat-sen University, Guangzhou 510275, China (e-mail: huangjj76@mail2.sysu.edu.cn).},
Chang-An Zhao\thanks{C.-A. Zhao is with the School of Mathematics, Sun Yat-sen University, Guangzhou 510275, China, and also with the Guangdong Key Laboratory of Information Security Technology, Guangzhou 510006, China (e-mail: zhaochan3@mail.sysu.edu.cn).}}

\maketitle

\begin{abstract}
Locally repairable codes are widely applicable in contemporary large-scale distributed cloud storage systems and various other areas. By making use of some algebraic structures of elliptic curves, Li et al. developed a series of $q$-ary optimal locally repairable codes with lengths that can extend to $q+2\sqrt{q}$. In this paper, we generalize their methods to hyperelliptic curves of genus $2$, resulting in the construction of several new families of $q$-ary optimal or almost optimal locally repairable codes. Our codes feature lengths that can approach $q+4\sqrt{q}$, and the locality can reach up to $239$.

\end{abstract}

\section{Introduction}

The introduction of locally repairable codes (LRCs) by Gopalan et al. \cite{On_the_Locality_of_Codeword_Symbols} marked a pivotal advancement in the domain of error-correcting codes, especially concerning their application to contemporary large-scale distributed cloud storage systems. In such environments, employing locally repairable codes facilitates the efficient restoration of data following node failures, a capability that is essential for preserving the reliability and availability of stored information. By constraining the number of nodes required to recover a single failed node, locally repairable codes effectively reduce both the bandwidth and computational demands involved in the data recovery process. Let $\mathbb{F}_q$ denote the finite field with $q$ elements. A $q$-ary $[n,k,d]_q$ code $\mathcal{C}$ is a linear subspace of $\mathbb{F}_q^n$ with dimension $k$ and minimum distance $d$. Codes that incorporate locality constraints are specifically constructed to endure the erasure of symbols within a codeword through the embedding of minimal redundancy in each codeword. The concept of locally repairable codes with locality $r$ is formally defined as follows\cite[Def. 4.1]{Locally_Recoverable_Codes_from_Algebraic_Curves_and_Surfaces}.

\begin{definition}
Let $\mathcal{C}\subseteq \mathbb{F}_q^n$ be a q-ary code. Given $a\in \mathbb{F}_q$, consider the sets of codewords
\[
    \mathcal{C}(i,a)=\{x\in \mathcal{C}:x_i=a\},\quad i\in\{1,\,2,\,\cdots,\,n\}.
\]
The code $\mathcal{C}$ is a locally repairable code with locality r if for every $i\in \{1,\,2,\,\cdots,\,n\}$ there exists an r-element subsets $I_i\subseteq \{1,\,2,\,\cdots,\,n\}\backslash \{i\}$ such that the restrictions of the sets $\mathcal{C}(i,a)$ to the coordinates in $I_i$ for different a are disjoint, i.e.
\[
    \mathcal{C}_{I_i}(i,a)\cap\mathcal{C}_{I_i}(i,a^\prime)=\emptyset,\quad whenever \quad a\neq a^\prime.
\]
\end{definition}

Large minimum distance $d$ and high rate $k/n$ are two desirable parameters for codes. For a $q$-ary $[n,k,d]_q$ locally repairable code $\mathcal{C}$ with locality $r$, it is proved in \cite{A_Family_of_Optimal_Locally_Recoverable_Codes} that the minimum distance $d$ of $\mathcal{C}$ satisfies the inequality
\begin{equation}\label{min_d_bound}
    d\le n-k-\bigg\lceil\frac{k}{r}\bigg\rceil+2.
\end{equation}
The bound (\ref{min_d_bound}) of the minimum distance of $\mathcal{C}$ is called the {\it Singleton-type upper bound}. The difference between the two terms in (\ref{min_d_bound}), $\Delta=n-k-\big\lceil\frac{k}{r}\big\rceil+2-d$, is the {\it Singleton-optimal defect} of $\mathcal{C}$. A code $\mathcal{C}$ with $\Delta=0\,(\Delta=1)$ is called an (almost) optimal locally repairable code. 

In \cite{A_Family_of_Optimal_Locally_Recoverable_Codes}, Tamo and Barg introduced a variation of Reed–Solomon codes to achieve local recoverability. These so-called locally repairable RS codes are optimal and have significantly lower locality than the RS codes themselves. However, their length is still smaller than the size of $\mathbb{F}_q$. The problem of finding optimal locally repairable codes when $n>q$ has always been a research hotspot. A classical approach to obtain longer codes is to use algebraic curves with many rational points. In this way, Barg, Tamo, and Vladut \cite{Locally_Recoverable_Codes_on_Algebraic_Curves} extended locally repairable RS codes to so-called locally repairable algebraic geometry codes, which in effect resulted in more locally repairable codes. The natural progression from these successes was the exploration of locally repairable codes using algebraic geometry methods, as indicated in \cite{Locally_Recoverable_Codes_from_Algebraic_Curves_and_Surfaces}, \cite{Optimal_Locally_Repairable_Codes_Via_Elliptic_Curves}, \cite{Locally_Recoverable_J-affine_variety_codes}, \cite{Locally_Recoverable_codes_with_availability_<i>t</i>≥2_from_fiber_products_of_curves}, \cite{Locally_Recoverable_codes_from_rational_maps}, \cite{Locally_Recoverable_codes_from_algebraic_curves_with_separated_variables},\cite{Construction_of_Optimal_Locally_Repairable_Codes_via_Automorphism_Groups_of_Rational_Function_Fields}. In particular, Li et al. provided a systematic method for producing optimal locally repairable codes via elliptic curves \cite{Optimal_Locally_Repairable_Codes_Via_Elliptic_Curves} and Jin et al. constructed optimal locally repairable codes by employing automorphism groups of rational function fields \cite{Construction_of_Optimal_Locally_Repairable_Codes_via_Automorphism_Groups_of_Rational_Function_Fields}. Apart from these, the algebraic curves most commonly considered for constructing locally repairable codes are generalized Giulietti and Korchmaros (GK) curves, Suzuki curves, Hermitian curves, Garcia-Stichtenoth curves and others \cite{Locally_Recoverable_J-affine_variety_codes,Locally_Recoverable_codes_with_availability_<i>t</i>≥2_from_fiber_products_of_curves,Locally_Recoverable_codes_from_rational_maps,Locally_Recoverable_codes_from_algebraic_curves_with_separated_variables}. In addition, there are several constructions of a locally repairable code with $\Delta=1$. In \cite{Constructions_of_optimal_and_almost_optimal_locally_repairable_codes} Ernvall et al. gave a construction for almost optimal locally repairable codes with parameters $(n,k,r)$ satisfying that $q>2\binom{n}{k-1}$ and $n\not\equiv 1\bmod{r+1}$. In \cite{Local_codes_with_addition_based_repair}, Han Mao Kiah et al. constructed an almost optimal locally repairable code with addition based repair and information locality $r$ where the parameters satisfy $n<q$ and $r\mid k$.

In this paper, we introduce a framework for constructing locally repairable codes derived from hyperelliptic curves with genus $2$, which, to a certain extent, generalize the methods discussed in \cite{Optimal_Locally_Repairable_Codes_Via_Elliptic_Curves}. The main idea is to consider a subgroup $\mathcal{G}$ of the automorphism group of the hyperelliptic curve $\mathfrak{C}$ such that the fixed field $F^{\mathcal{G}}$ of the function field $F=\mathbb{F}_q(\mathfrak{C})$, with respect to $\mathcal{G}$, is a rational function field. Given several suitable conditions, it is feasible to identify $|\mathcal{G}|-1$ elements within a specific Riemann-Roch space that ensure local recoverability. Subsequently, we can design the size of the dimension $k$ such that the locally repairable code that we constructed is either optimal or almost optimal. The length of our constructed code can exceed $q$ and approach $q+4\sqrt{q}$, while the locality $r$ can be as big as $239$. In Table \ref{opt_parameter}, we present a comparison of the parameters of our codes with those of earlier constructions, which were constructed via algebraic curves and have specific Singleton-optimal defect. 
    \begin{table}[ht]
	\centering
	\begin{tabular}{ccccc}
        \hline
			Ref. & Length $n$ & Dimension $k$ & Locality $r$ & Singleton-optimal defect $\Delta$ \\
        \hline
         \cite{Optimal_Locally_Repairable_Codes_Via_Elliptic_Curves} & $\ell(r+1)$ & $rt-(r-1)$ & $3,5,7,11,23$ & $0$ \\
         
          \cite{Optimal_Locally_Repairable_Codes_Via_Elliptic_Curves} & $3\ell$ & $2t+1$ & $2$ & $0$ \\

         \cite[Cor. 3(a)]{Locally_Recoverable_codes_from_algebraic_curves_with_separated_variables} & $\le u\cdot v(\phi_1)$ & \makecell{$k^\ast$ \\ the dimension of $\mathcal{C}(\mathcal{P},\mathcal{L}(mQ))$} & $\le \min\{r_1,\cdots,r_u\}$ & $\le g+1-\lceil(m+1-g)/(b-1)\rceil$ \\

         \cite[Cor. 3(b)]{Locally_Recoverable_codes_from_algebraic_curves_with_separated_variables} & $\le u\cdot v(\phi_1)$ & \makecell{$k^\ast-1-\ell_{b-1}$ \\ $\ell_{b-1}=\lfloor(m-a(b-1))/b\rfloor$} & $\le \min\{r_1,\cdots,r_u\}$ & \makecell{$\le g+2+\ell_{b-1}$ \\ $-\lceil(m-g-\ell_{b-1})/(b-1)\rceil$} \\

          \cite[Sec. IV]{Construction_of_Optimal_Locally_Repairable_Codes_via_Automorphism_Groups_of_Rational_Function_Fields} & $m(r+1)$ & $rt$ & \makecell{$p^v-1,up^v-1$ \\ or $(r+1)\mid (q-1)$} & 0 \\

           \cite[Thm. V.3]{Construction_of_Optimal_Locally_Repairable_Codes_via_Automorphism_Groups_of_Rational_Function_Fields} & $m(r+1)$ & $rt$ & $(r+1)\mid(q+1)$ & 0 \\
          
	   Thm. \ref{r_odd} & $\ell(r+1)$ & $rt-(r-1)$ & \makecell{$|\mathcal{G}|-1$ \\ $|\mathcal{G}|$ is even} & $\le1$ \\
    
        Thm. \ref{r_even},\,Rem. \ref{r_even_re}(2) & $5\ell$ & $4t+1$ or $4t+2$ & $4$ & $\le1$\\
        
        Rem. \ref{r_even_re}(1) & $3\ell$ & $2t+1$ & $2$ & $\le1$\\
        \hline
	\end{tabular}
	\caption{The parameters of some locally repairable codes}
	\label{opt_parameter}
    \end{table}


The paper is organized as follows. Section \ref{pre} provides some preliminaries about hyperelliptic curves, theory of function fields and so on. In the subsequent section, we will divide this into three parts, and begin with a general framework for constructing locally repairable codes via automorphism groups of hyperelliptic curves with genus 2. Following that, we will discuss two specific constructions, one for the case where locality $r$ is odd (Theorem \ref{r_odd}) and the other for even (Theorem \ref{r_even}). Each case is illustrated with several examples. In Section \ref{conclusion}, we give a summary of the entire article.

\section{Preliminaries}\label{pre}
In this section, we provide a concise review of some preliminaries related to hyperelliptic curves of genus 2 over finite fields, algebraic geometry codes, theory of function fields and automorphism groups of hyperelliptic curves with genus 2.

\subsection{Hyperelliptic Curves of Genus 2 Over Finite Fields}

Let $q$ be a power of an odd prime $p$. Let $\mathfrak{C}$ be a projective, non-singular, geometrically irreducible algebraic curve with genus 2 defined over $\mathbb{F}_q$. In general, we will implicitly assume that curves of genus 2 over $\mathbb{F}_q$ are given by a hyperelliptic model
\begin{equation}\label{hyper_eqa}
    \mathfrak{C}:\quad y^2=f(x)
\end{equation}
with $f(x)$ of degree 5 or 6 without multiple roots. Such a curve has a unique singularity at the point at infinity that corresponding to one or two points in a hyperelliptic model, depending on whether the degree of polynomial $f(x)$ is odd or even. See \cite{Prolegomena_to_a_Middlebrow_Arithmetic_of_Curves_of_Genus_2,Construction_de_courbes_de_genre_2_`a_partir_de_leurs_modules,Computational_aspects_of_curves_of_genus_at_least_2} for more details about hyperelliptic curves and curves of genus 2. We write $\mathfrak{C}/\mathbb{F}_q$ as a hyperelliptic curve $\mathfrak{C}$ defined over $\mathbb{F}_q$. 

Denote by $F=\mathbb{F}_q(\mathfrak{C})$ and $\mathfrak{C}(\mathbb{F}_q)$ the function field of $\mathfrak{C}/\mathbb{F}_q$ and the set of $\mathbb{F}_q$-rational points, respectively. Then the function field $F$ is given by $F=\mathbb{F}_q(x,y)$, where $x$ and $y$ satisfy the equation (\ref{hyper_eqa}). Let $\mathbb{P}_F$ be the set of all places of $F$ and $\mathbb{P}_F^1=\{P\in\mathbb{P}_F:\deg P =1\}$ be the set of rational places of $F$. There is a one-to-one correspondence between $\mathfrak{C}(\mathbb{F}_q)$ and $\mathbb{P}_F^1$. More specifically, the rational point $(a,b)$ on $\mathfrak{C}$ corresponds to the unique common zero of $x-a$ and $y-b$, denoted by $P_{a,b}$. The point $(a,-b)$ also lies on $\mathfrak{C}$ which corresponds to the unique common zero of $x-a$ and $y+b$, and we will denote it by $\bar{P}_{a,b}$. If $\mathfrak{C}$ has a unique point at infinity $\infty$, then we denote the place corresponding to it by $P_{\infty}$; if $\mathfrak{C}$ has two points at infinity $\infty^+$ and $\infty^-$, then we denote the places corresponding to them by $P_{\infty^+}$ and $P_{\infty^-}$ when $2\mid [\mathbb{F}_q:\mathbb{F}_p]$.

The following lemma gives a bound on the size of $\mathbb{P}_F^1$ which is called the Hasse-Weil Bound \cite[Thm. 5.2.3]{Algebraic_Function_Fields_and_Codes}.

\begin{lemma}
    The number $\#\mathbb{P}_F^1$ of rational places of $F/\mathbb{F}_q$ with genus $2$ satisfies the inequality
    \begin{equation}\label{Hasse-Weil_Bound}
        |\#\mathbb{P}_F^1-(q+1)|\le 4q^{1/2}.
    \end{equation}
\end{lemma}

If the number $\#\mathbb{P}_F^1$ attains the upper bound $q+1+4q^{1/2}$, then $\mathfrak{C}$ is called a maximal hyperelliptic curve over $\mathbb{F}_q$. Next, we introduce two types of maximal hyperelliptic curves with genus 2\cite[Thm. 1 and Thm. 6]{A_note_on_certain_maximal_hyperelliptic_curves}.

\begin{lemma}\label{max_hy}
    \begin{itemize}
        \item[(i)] The hyperelliptic curve $\mathfrak{C}$ of genus $2$ corresponding to 
        \[
            y^2=x^5+x
        \]
        is maximal over $\mathbb{F}_{q^2}$ if and only if $q\equiv 5 $ or $7 \pmod{8}$.
        \item[(ii)] The hyperelliptic curve $\mathfrak{C}$ of genus $2$ corresponding to 
        \[
            y^2=x^5+1
        \]
        is maximal over $\mathbb{F}_{q^2}$ if and only if $5$ divides $q+1$.
    \end{itemize}
\end{lemma}

In the following lemma, we recall the maximality of a curve over a constant field extension\cite[Prop. 2]{A_note_on_certain_maximal_hyperelliptic_curves}. 

\begin{lemma}\label{maximality}
    Let $\mathfrak{C}$ be a maximal curve over $\mathbb{F}_{q^2}$. Then $\mathfrak{C}$ is maximal over the constant field extension $\mathbb{F}_{q^{2s}}$ if and only if $s$ is odd. 
\end{lemma}

The divisor group of $F/\mathbb{F}_q$ is defined as the free abelian group which is generated by $\mathbb{P}_F$; it is denoted by $\text{Div}(F)$. Two divisors $D,D^\prime\in\text{Div}(F)$ are said to be linearly equivalent, written $D\sim D^\prime$, if $D=D^\prime+(z)^{F}$ for some $z\in F\backslash\{0\}$. We denote the divisor class group of $\mathfrak{C}$ by $\text{Cl}(F)$, which is defined as
\[
    \text{Cl}(F)=\text{Div}(F)/\sim.
\]
The class of a divisor $D$ in $\text{Cl}(F)$ will be denoted by $[D]$. We define $\text{Div}^0(F)$ as the degree zero subgroup of $\text{Div}(F)$ and define $\text{Cl}^0(F)$ as the degree zero subgroup of $\text{Cl}(F)$. To describe elements of $\text{Cl}^0(F)$ we will need an effective divisor $D_{\infty}$, and this divisor will be given as follow\cite[Def. 3]{Efficient_Hyperelliptic_Arithmetic_Using_Balanced_Representation_for_Divisors}.

\begin{definition}
    \begin{itemize}
        \item If $\mathfrak{C}$ has a unique point at infinity $\infty$, then $D_\infty=2P_\infty$.
        \item If $\mathfrak{C}$ has two points at infinity $\infty^+$ and $\infty^-$, then $D_\infty=P_{\infty^+}+P_{\infty^-}$.
    \end{itemize}
\end{definition}

The following lemma gives a description of the elements of $\text{Cl}^0(F)$\cite[Prop. 1]{Efficient_Hyperelliptic_Arithmetic_Using_Balanced_Representation_for_Divisors}.

\begin{lemma}\label{unique}
    Let $D_\infty$ be defined as above, and let $D\in{\rm Div}^0(F)$. Then $[D]$ has a unique representative in ${\rm{Cl}}^0(F)$ of the form $[D_0-D_\infty]$, where $D_0=\sum_{i=1}^2 P_i\in {\rm{Div}}(F)$ is an effective divisor of degree 2 whose affine part satisfies $P_i\neq \bar{P}_j$ for all $i\neq j$.
\end{lemma}

\subsection{Algebraic Geometry Codes}
For more details of algebraic geometry codes, the reader may refer to \cite{Algebraic-Geometric_Codes}. Let $F/\mathbb{F}_q$ be a function field of genus $g$ with the full constant field $\mathbb{F}_q$. Let $D=\{P_1,\cdots,P_n\}\subseteq\mathbb{P}_F^1$ be a set of $n$ distinct rational places of $F$. For a divisor $G$ of $F/\mathbb{F}_q$ with $2g-2<\deg G<n$ and ${\rm{supp}}(G)\cap D=\emptyset$, the algebraic geometry code associated with the divisors $D$ and $G$ is defined as
\[
    C_{\mathcal{L}}(D,G):=\{(x(P_1),\,\cdots,\,x(P_n)):x\in\mathcal{L}(G)\}\subseteq\mathbb{F}_q^n,
\]
where $\mathcal{L}(G)$ is the Riemann-Roch space with the dimension $\dim_{\mathbb{F}_q}\mathcal{L}(G)=\deg G+1-g$ from the Riemann-Roch Theorem \cite[Thm. 1.5.15]{Algebraic_Function_Fields_and_Codes}. Then the code $C_{\mathcal{L}}(D,G)$ is an $[n,k,d]_q$ linear code with dimension $k=\dim_{\mathbb{F}_q}\mathcal{L}(G)$ and minimum distance $d\ge n-\deg G$. If $V$ is a subspace of $\mathcal{L}(G)$, then we can define a subcode of $C_{\mathcal{L}}(D,G)$ by 
\[
    C_{\mathcal{L}}(D,V):=\{(x(P_1),\,\cdots,\,x(P_n)):x\in V\}.
\]
Then the dimension of $C_{\mathcal{L}}(D,V)$ is the dimension of the space $V$ over $\mathbb{F}_q$ and the minimum distance of $C_{\mathcal{L}}(D,V)$ is still lower bounded by $n-\deg G$.

In \cite{Construction_of_Optimal_Locally_Repairable_Codes_via_Automorphism_Groups_of_Rational_Function_Fields}, Jin et al. have modified the above construction. Keep the notations above and let $m_i=v_{P_i}(G)$. Choose a local parameter $\pi_{P_i}$ of $P_i$ for each $i\in\{1,\,\cdots,\,n\}$. Then for any nonzero $x\in \mathcal{L}(G)$, we have $v_{P_i}(\pi_{P_i}^{m_i}x)=m_i+v_{P_i}(x)\ge m_i-v_{P_i}(G)=0$. Define a modified algebraic geometry code as follows
\[
     C_{\mathcal{L}}(D,G):=\{((\pi_{P_1}^{m_1}x)(P_1),\,\cdots,\,(\pi_{P_n}^{m_n}x)(P_n)):x\in\mathcal{L}(G)\}.
\]
Then $C_{\mathcal{L}}(D,G)$ is still an $[n,\deg G+1-g,\ge n-\deg G]_q$ linear code. Now for a subspace $V$ of $\mathcal{L}(G)$, we can define a subcode of $C_{\mathcal{L}}(D,G)$ by 
\[
    C_{\mathcal{L}}(D,V):=\{((\pi_1^{m_1}x)(P_1),\,\cdots,\,(\pi_n^{m_n}x)(P_n)):x\in V\}.
\]
Again, the dimension of $C_{\mathcal{L}}(D,V)$ is the dimension of the space $V$ over $\mathbb{F}_q$ and the minimum distance of $C_{\mathcal{L}}(D,V)$ is lower bounded by $n-\deg G $.

\subsection{Theory of Function Fields}
Let $F/\mathbb{F}_q$ be a function field of genus $g(F)$ with the full constant field $\mathbb{F}_q$. Let $\mathbb{P}_F$ denote the set of places of $F$ and let ${\rm{Aut}}(F/\mathbb{F}_q)$ be the automorphism group of $F$ over $\mathbb{F}_q$, i.e.
\[
    {\rm{Aut}}(F/\mathbb{F}_q)=\{\sigma:\sigma \text{ is an } \mathbb{F}_q\text{-automorphism of } E\}.
\]
Now let $\mathcal{G}$ be a finite subgroup of ${\rm{Aut}}(F/\mathbb{F}_q)$. The fixed subfield of $F$ with respect to $\mathcal{G}$ is defined by
\[
    F^{\mathcal{G}}=\{z\in F:\sigma(z)=z \text{ for all } \sigma\in\mathcal{G}\}.
\]
From the Galois theory, $F/F^{\mathcal{G}}$ is a Galois extension with ${\rm{Gal}}(F/F^{\mathcal{G}})=\mathcal{G}$. Moreover, $F^{\mathcal{G}}/\mathbb{F}_q$ is also a function field with the full constant field $\mathbb{F}_q$. By \cite[Lem. 3.5.2]{Algebraic_Function_Fields_and_Codes}, for any automorphism $\sigma\in{\rm{Gal}}(F/F^{\mathcal{G}})$ and any place $P\in\mathbb{P}_F$, then $\sigma(P):=\{\sigma(z):z\in P\}$ is a place of $F$ as well. Let $g(F^{\mathcal{G}})$ denote the genus of $F^{\mathcal{G}}$. Then the Hurwitz Genus Formula\cite[Thm. 3.4.13]{Algebraic_Function_Fields_and_Codes} yields
\begin{equation}\label{Hurwitz_genus_formula}
    2g(F)-2=[F:F^{\mathcal{G}}](2g(F^{\mathcal{G}})-2)+\deg {\rm{Diff}}(F/F^{\mathcal{G}}),
\end{equation}
where ${\rm{Diff}}(F/F^{\mathcal{G}})$ stands for the different of $F/F^{\mathcal{G}}$.

\subsection{Automorphism Groups of Hyperelliptic Curves with Genus 2}
The isomorphisms between two hyperelliptic curves correspond, in terms of hyperelliptic models, to transformations of the type
\begin{equation}\label{trans_eq}
    x^\prime=\frac{ax+b}{cx+d},\quad y^\prime=\frac{(ad-bc)y}{(cx+d)^3},
\end{equation}
associated to a uniquely determined matrix
\[
    M=\begin{pmatrix}
a & b \\
c & d
\end{pmatrix}\in{\rm{GL}}_2(\overline{\mathbb{F}}_q).
\]
The field of definition of that isomorphism is the field generated, over the common field of definition of the two curves, by the coefficients of the matrix $M$. In particular, for every curve $\mathfrak{C}/\mathbb{F}_q$ of genus 2, and once a hyperelliptic model for $\mathfrak{C}$ is fixed, the group of automorphisms ${\rm{Aut}}(\mathfrak{C})={\rm{Aut}}_{\overline{\mathbb{F}}_q}(\mathfrak{C})$ can be identified with a subgroup of ${\rm{GL}}_2(\overline{\mathbb{F}}_q)$ which is closed by the Galois action of the group $G_{\mathbb{F}_q}$. Here $G_{\mathbb{F}_q}$ is the Galois group of an algebraic closure $\overline{\mathbb{F}}_q/\mathbb{F}_q$. Denote ${\rm{Aut}}(\mathfrak{C}/\mathbb{F}_q)$ the subgroup of ${\rm{Aut}}(\mathfrak{C})$ in which every automorphism is defined over $\mathbb{F}_q$. 

In \cite{On_Binary_Sextics_with_Linear_Transformations_into_Themselves}, Bolza gives the different possibilities for the reduced group of automorphisms ${\rm{Aut}}^\prime(\mathfrak{C})={\rm{Aut}}(\mathfrak{C})/\langle\imath\rangle$ of a genus 2 curve, where $\imath$ is the hyperelliptic involution. The corresponding structures for the full group ${\rm{Aut}}(\mathfrak{C})$ are given in \cite{On_curves_of_genus_2_with_Jacobian_of_GL_2-type}. The picture, outside of characteristics 2, 3 and 5, is the following: the group ${\rm{Aut}}(\mathfrak{C})$ is isomorphic to one of the groups
\[
    C_2,\;V_4,\;D_8,\;D_{12},\;2D_{12},\;\tilde{S}_4,\;C_{10}.
\]
Here $C_n$ denotes the cyclic group of order $n$, $V_4$ is the Klein 4-group, $D_n$ is the dihedral group of order $n$, and $2D_{12}$, $\tilde{S}_4$ are certain 2-coverings of the dihedral and symmetric groups $D_{12}$ and $S_4$, respectively.  

Cardona studies the two families of curves of genus 2 with ${\rm{Aut}}(\mathfrak{C})\simeq D_8$ or $D_{12}$ in \cite{Curves_of_genus_2_with_group_of_automorphisms_isomorphic_to_D_8_or_D12} and genus 2 curves with ${\rm{Aut}}(\mathfrak{C})\simeq V_4$ over an algebraically closed field in \cite{On_the_number_of_curves_of_genus_2_over_a_finite_field}. The specific structures of the automorphism groups ${\rm{Aut}}(\mathfrak{C})$ of many genus 2 curves are also given in \cite{On_the_number_of_curves_of_genus_2_over_a_finite_field}. Now we focus on certain specific hyperelliptic curves and determine their automorphism groups ${\rm{Aut}}(\mathfrak{C}/\mathbb{F}_q)$. 

\begin{lemma}\label{auto_non}
    \begin{itemize}
        \item[(i)] Let $q$ be a power of an odd prime and let $\mathfrak{C}$ be an hyperelliptic curve over $\mathbb{F}_q$ defined by the equation $y^2=x^5+x^3+tx$ for some $t\in \mathbb{F}_q\setminus\{0,1/4,9/100\}$. If $4\mid q-1$ and $t^{1/4}\in\mathbb{F}_q$, then ${\rm{Aut}}(\mathfrak{C}/\mathbb{F}_q)\simeq D_8$. 
        \item[(ii)] Let $q$ be a power of an odd prime other than 3, and let $\mathfrak{C}$ be an hyperelliptic curve over $\mathbb{F}_q$ defined by the equation $y^2=x^6+x^3+t$ for some $t\in \mathbb{F}_q\setminus\{0,1/4,-1/50\}$. If $3\mid q-1$ and $t^{1/6}\in\mathbb{F}_q$, then ${\rm{Aut}}(\mathfrak{C}/\mathbb{F}_q)\simeq D_{12}$. 
    \end{itemize}
\end{lemma}
\begin{proof}
    (i) For $t\in \mathbb{F}_q\setminus\{0,1/4,9/100\}$, let
    \[
        U=\begin{pmatrix}
        -(-1)^{1/2} & 0 \\
        0 & (-1)^{1/2}
        \end{pmatrix}\quad {\rm and} \quad V=\begin{pmatrix}
        0 & t^{1/4} \\
        t^{-1/4} & 0
        \end{pmatrix}.
    \]
    By \cite[Prop. 2.1]{Curves_of_genus_2_with_group_of_automorphisms_isomorphic_to_D_8_or_D12}, the automorphism group of $\mathfrak{C}:y^2=x^5+x^3+tx$ is
    \[
        {\rm{Aut}}(\mathfrak{C})=\langle \,U,\,V\,\rangle\simeq D_8.
    \]
    Since $4\mid q-1$ and $t^{1/4}\in\mathbb{F}_q$, we have $(-1)^{1/2}\in\mathbb{F}_q$ and then $U,V\in {\rm{GL}}_2(\mathbb{F}_q)$. Thus 
    \[
        {\rm{Aut}}(\mathfrak{C}/\mathbb{F}_q)={\rm{Aut}}(\mathfrak{C})\simeq D_8. 
    \]
    
    (ii) By \cite[Prop. 2.2]{Curves_of_genus_2_with_group_of_automorphisms_isomorphic_to_D_8_or_D12}, the automorphism group of $\mathfrak{C}:y^2=x^6+x^3+t$ is
    \[
        {\rm{Aut}}(\mathfrak{C})\simeq D_{12},
    \]
    for $t\in \mathbb{F}_q\setminus\{0,1/4,-1/50\}$. Let
    \[
        U=\begin{pmatrix}
        -\alpha^2 & 0 \\
        0 & -\alpha
        \end{pmatrix}\quad {\rm and} \quad V=\begin{pmatrix}
        0 & t^{1/6} \\
        t^{-1/6} & 0
        \end{pmatrix},
    \]
    where ${\rm ord}(\alpha)=3$. Let $\sigma_U,\sigma_V$ be the isomorphisms corresponding to $U,V$ respectively. By (\ref{trans_eq}), we have
    \[
        \sigma_U(x)=\alpha x,\quad \sigma_U(y)=-y\quad {\rm and}\quad \sigma_V(x)=\frac{t^{1/3}}{x},\quad \sigma_V(y)=\frac{-y}{t^{-1/2}x^3}.
    \]
    Then we can verify that $\sigma_U,\sigma_V$ are indeed the automorphisms of $\mathfrak{C}$. Moreover, by calculation, 
    \[
        \langle\,U,\,V\mid U^6=V^2=I,\,V^{-1}UV=U^{-1}\rangle\simeq D_{12}.
    \]
    Thus 
    \[
        {\rm{Aut}}(\mathfrak{C})=\langle\,U,\,V\,\rangle.
    \]
    Since $3\mid q-1$ and $t^{1/6}\in\mathbb{F}_q$, we have $\alpha\in\mathbb{F}_q$ and then $U,V\in {\rm{GL}}_2(\mathbb{F}_q)$. Hence 
    \[
        {\rm{Aut}}(\mathfrak{C}/\mathbb{F}_q)={\rm{Aut}}(\mathfrak{C})=\langle\,U,\,V\,\rangle\simeq D_{12}. 
    \]
    The proof is completed. 
\end{proof}

\begin{lemma}\label{auto_max_1}
    Let $q$ be a power of an odd prime such that $8\mid q-1$ and $2^{1/2}\in\mathbb{F}_q$. Let $\mathfrak{C}$ be an hyperelliptic curve over $\mathbb{F}_q$ defined by the equation $y^2=x^5+x$. 
    \begin{itemize}
        \item[(i)]  If ${\rm char}\; \mathbb{F}_q\neq3$, $5$, then ${\rm{Aut}}(\mathfrak{C}/\mathbb{F}_q)\simeq \tilde{S}_4$. 
        \item[(ii)] If ${\rm char}\; \mathbb{F}_q=5$, then ${\rm{Aut}}(\mathfrak{C}/\mathbb{F}_q)\simeq \tilde{S}_5$. 
    \end{itemize}
\end{lemma}
\begin{proof}
    Let $\mathfrak{C}^\prime$ be an hyperelliptic curve over $\mathbb{F}_q$ defined by the equation $y^2=x^5-x$. Let 
    \[
        T=\begin{pmatrix}
        0 & (-1)^{1/8} \\
        -(-1)^{3/8} & 0
        \end{pmatrix}.
    \]
    Let $\sigma_T$ be the isomorphism corresponding to $T$. By (\ref{trans_eq}), we have
    \[
        \sigma_T(x)=\frac{-1}{(-1)^{1/4}x},\quad \sigma_T(y)=\frac{(-1)^{1/2}y}{(-1)^{1/8}x^3}.
    \]
    Then we can verify that $\sigma_T$ is the isomorphism from $\mathfrak{C}^\prime$ to $\mathfrak{C}$. 

    (i) Let 
    \[
        U=2^{-1/2}\begin{pmatrix}
        -1 & -(-1)^{1/2} \\
        (-1)^{1/2} & 1
        \end{pmatrix}\quad {\rm and} \quad V=2^{-1/2}\begin{pmatrix}
        (-1)^{1/2}+1 & 0 \\
        0 & (-1)^{1/2}-1
        \end{pmatrix}.
    \]
    By Tables 1 and 2 in \cite{On_curves_of_genus_2_with_Jacobian_of_GL_2-type}, 
    \[
        {\rm{Aut}}(\mathfrak{C}^\prime)=\langle\,U,\,V\,\rangle\simeq \tilde{S}_4. 
    \]
    Therefore, we can know that the transformations corresponding to $U^\prime=TUT^{-1}$ and $V^\prime=TVT^{-1}$ are the automorphisms of $\mathfrak{C}$. Moreover, 
    \[
        {\rm{Aut}}(\mathfrak{C})=\langle\,U^\prime,\,V^\prime\,\rangle\simeq \tilde{S}_4.
    \]
    By calculation, we have
    \[
        U^\prime=2^{-1/2}\begin{pmatrix}
        1 & -(-1)^{1/4} \\
        (-1)^{3/4} & -1
        \end{pmatrix}\quad {\rm and} \quad V^\prime=2^{-1/2}\begin{pmatrix}
        (-1)^{1/2}-1 & 0 \\
        0 & (-1)^{1/2}+1
        \end{pmatrix}.
    \]
    Since $8\mid q-1$ and $2^{1/2}\in\mathbb{F}_q$, we have $(-1)^{1/4}\in\mathbb{F}_q$ and then $U^\prime,V^\prime\in {\rm{GL}}_2(\mathbb{F}_q)$. Thus 
    \[
        {\rm{Aut}}(\mathfrak{C}/\mathbb{F}_q)={\rm{Aut}}(\mathfrak{C})=\langle\,U^\prime,\,V^\prime\,\rangle\simeq \tilde{S}_4. 
    \]

    (ii) Let 
    \[
        U=\begin{pmatrix}
        0 & 2 \\
        2 & 0
        \end{pmatrix},\quad V=\begin{pmatrix}
        1 & 2 \\
        2 & 0
        \end{pmatrix}\quad {\rm and} \quad W=2^{1/2}\begin{pmatrix}
        2 & 0 \\
        0 & 1
        \end{pmatrix}.
    \]
    By section 2.3.2 in \cite{On_the_number_of_curves_of_genus_2_over_a_finite_field}, 
    \[
        {\rm{Aut}}(\mathfrak{C}^\prime)=\langle\,U,\,V,\,W\,\rangle\simeq \tilde{S}_5. 
    \]
    Similar to (i), we can derive that 
    \[
        {\rm{Aut}}(\mathfrak{C})=\langle\,U^\prime,\,V^\prime,\,W^\prime\,\rangle\simeq \tilde{S}_5,
    \]
    where $U^\prime=TUT^{-1}$, $V^\prime=TVT^{-1}$ and $W^\prime=TWT^{-1}$. More precisely, 
    \[
        U^\prime=\begin{pmatrix}
        0 & -(-1)^{-1/4}\cdot2 \\
        -(-1)^{1/4}\cdot2 & 0
        \end{pmatrix},\enspace V^\prime=\begin{pmatrix}
        0 & -(-1)^{-1/4}\cdot2 \\
        -(-1)^{1/4}\cdot2 & 1
        \end{pmatrix}\enspace {\rm and} \enspace W^\prime=2^{1/2}\begin{pmatrix}
        1 & 0 \\
        0 & 2
        \end{pmatrix}.
    \]
    Since $8\mid q-1$ and $2^{1/2}\in\mathbb{F}_q$, we have $(-1)^{1/4}\in\mathbb{F}_q$ and then $U^\prime,V^\prime,W^\prime\in {\rm{GL}}_2(\mathbb{F}_q)$. Thus 
    \[
        {\rm{Aut}}(\mathfrak{C}/\mathbb{F}_q)={\rm{Aut}}(\mathfrak{C})=\langle\,U^\prime,\,V^\prime,\,W^\prime\,\rangle\simeq \tilde{S}_5.
    \]
    The proof is completed. 
\end{proof}

\begin{lemma}\label{auto_max_2}
    Let $q$ be a power of an odd prime other than 5 and let $\mathfrak{C}$ be an hyperelliptic curve over $\mathbb{F}_q$ defined by the equation $y^2=x^5+1$. If $5\mid q-1$, then ${\rm{Aut}}(\mathfrak{C}/\mathbb{F}_q)\simeq C_{10}$. 
\end{lemma}
\begin{proof}
    Let $\mathfrak{C}^\prime$ be an hyperelliptic curve over $\mathbb{F}_q$ defined by the equation $y^2=x^5-1$. Let 
    \[
        T=\begin{pmatrix}
        (-1)^{1/10} & 0 \\
        0 & (-1)^{3/10}
        \end{pmatrix}.
    \]
    Similar to the proof of Lemma \ref{auto_max_1}, we can verify that the transformation corresponding to $T$ is the isomorphism from $\mathfrak{C}^\prime$ to $\mathfrak{C}$. Let 
    \[
        U=\begin{pmatrix}
        -1 & 0 \\
        0 & -1
        \end{pmatrix}\quad {\rm and} \quad V=\begin{pmatrix}
        \zeta^2 & 0 \\
        0 & \zeta
        \end{pmatrix},
    \]
    where ${\rm ord}(\zeta)=5$. By section 2.1.3 in \cite{On_the_number_of_curves_of_genus_2_over_a_finite_field}, 
    \[
        {\rm{Aut}}(\mathfrak{C}^\prime)=\langle\,U,\,V\,\rangle\simeq C_{10}. 
    \]
    Similar to the proof of Lemma \ref{auto_max_1}(i), we can derive that 
    \[
        {\rm{Aut}}(\mathfrak{C})=\langle\,U^\prime,\,V^\prime\,\rangle\simeq C_{10},
    \]
    where $U^\prime=TUT^{-1}$ and $V^\prime=TVT^{-1}$. By computation, we have $U^\prime=U$ and $V^\prime=V$. Hence, we conclude that
    \[
        {\rm{Aut}}(\mathfrak{C}/\mathbb{F}_q)={\rm{Aut}}(\mathfrak{C})=\langle\,U,\,V\,\rangle\simeq C_{10},
    \]
    by $5\mid q-1$. The proof is completed. 
\end{proof}

\section{Construction of locally repairable codes via hyperelliptic curves}\label{construction}

In this section, we initially provide a general framework for constructing locally repairable codes via automorphism groups of hyperelliptic curves with genus $2$. Following this, we present two specific constructions for the case where the locality $r$ is odd or even.

\subsection{A General Framework for Constructing LRCs via Automorphism Groups}

Let $\mathfrak{C}/\mathbb{F}_q$ be a hyperelliptic curve with function field $F$ and ${\rm{Aut}}(F/\mathbb{F}_q)$ be the automorphism group of $F$ over $\mathbb{F}_q$. Let $\mathcal{G}$ be a subgroup of ${\rm{Aut}}(F/\mathbb{F}_q)$ of order $r+1$ and $F^{\mathcal{G}}$ be the fixed subfield of $F$ with respect to $\mathcal{G}$. Then $F/F^{\mathcal{G}}$ is a Galois extension with Galois group ${\rm{Gal}}(F/F^{\mathcal{G}})=\mathcal{G}$.

Let $Q_1,Q_2,\cdots,Q_\ell$ be rational places of $F^{\mathcal{G}}$ which split completely in $F/F^{\mathcal{G}}$. Let $P_{i,1},P_{i,2},\cdots,P_{i,r+1}$ be the $r+1$ rational places of $F$ lying over $Q_i$ for $1\le i\le \ell$. For any $1\le i\le \ell$, set 
\[
    \mathcal{P}_i=\{P_{i,j}:\,1\le j\le r+1\},
\]
and set
\begin{equation}\label{ev_P}
    \mathcal{P}=\{P_{i,j}:1\le i\le \ell,\,1\le j\le r+1\}. 
\end{equation}

Let $e_1,e_2,\cdots,e_r$ be elements of $F$ that are linearly independent over $F^{\mathcal{G}}$ and $v_{P_{i,j}}(e_u)\ge0$ for $1\le i\le \ell,\,1\le j\le r+1$ and $1\le u\le r$. Let $f_1,f_2,\cdots,f_t$ be elements of $F^{\mathcal{G}}$ that are linearly independent over $\mathbb{F}_q$ and $v_{Q_i}(f_v)\ge0$ for $1\le i\le \ell$ and $1\le v\le t$. Define
\begin{equation}\label{func}
    V=\Big\{ \sum_{j=1}^ta_{1,j}f_je_1+\sum_{i=2}^r\big(\sum_{j=1}^{t-1}a_{i,j}f_j\big)e_i:a_{i,j}\in\mathbb{F}_q \Big\}.
\end{equation}

The following lemmas give local recoverability under certain conditions and algebraic geometry code with locality $r$ respectively.  

\begin{lemma}\label{local}\cite[Prop. 4.2]{Locally_Recoverable_Codes_from_Algebraic_Curves_and_Surfaces}
    Let $i$ be an integer between 1 and $s$, and suppose every $r\times r$ submatrix of the matrix 
    \[
        M=\begin{pmatrix}
        e_1(P_{i,1}) & e_2(P_{i,1}) & \cdots & e_r(P_{i,1}) \\
        e_1(P_{i,2}) & e_2(P_{i,2}) & \cdots & e_r(P_{i,2}) \\
        \vdots & \vdots & \ddots & \vdots \\
        e_1(P_{i,r+1}) & e_2(P_{i,r+1}) & \cdots & e_r(P_{i,r+1})
        \end{pmatrix}
    \]
    is invertible. Then the value of $f\in V$ at any place in the set $\mathcal{P}_i$ can be recovered from the values of $f$ at the other $r$ places of the set $\mathcal{P}_i$.
\end{lemma}

\begin{lemma}\label{cons}\cite[Prop. 19]{Optimal_Locally_Repairable_Codes_Via_Elliptic_Curves}
    Let $\mathcal{P}$ and $V$ be defined as above and satisfy the assumption of Lemma \ref{local}. If $V$ is contained in $\mathcal{L}(G)$ for a divisor $G$ of $F$ with $\deg G<\ell(r+1)$ and ${\rm supp}(G)\cap\mathcal{P}=\emptyset$, then the algebraic geometry code
    \[
        C_{\mathcal{L}}(\mathcal{P},V)=\{(f(P))_{P\in\mathcal{P}}:f\in V\}
    \]
    is a $q$-ary $[n,k,d]_q$-locally repairable code with locality $r$, length $n=\ell(r+1)$, dimension $k=rt-(r-1)$ and minimum distance $d\ge n-\deg G$. 
\end{lemma}

Now we focus on the automorphism group ${\rm{Aut}}(\mathfrak{C}/\mathbb{F}_q)$. By considering certain subgroups of ${\rm{Aut}}(\mathfrak{C}/\mathbb{F}_q)$, we can determine $F^{\mathcal{G}}$ and find $e_1,e_2,\cdots,e_r\in F$ that are linearly independent over $F^{\mathcal{G}}$.

\begin{proposition}\label{find_e_i}
    Let $\mathfrak{C}/\mathbb{F}_q$ be a hyperelliptic curve defined by the hyperelliptic model (\ref{hyper_eqa}) and $F=\mathbb{F}_q(x,y)$ be the function field $\mathbb{F}_q(\mathfrak{C})$. Let $\mathcal{G}$ be a subgroup of ${\rm{Aut}}(\mathfrak{C}/\mathbb{F}_q)$. Assume that 
    \[
        \#\mathcal{G}=r+1=2s \quad and \quad \#\{\sigma(x):\sigma\in\mathcal{G}\}=s
    \]
    for a positive integer $s\ge 2$ and $r<q$. Moreover, suppose that there exists a rational place $P_{a,b}\in \mathbb{P}_F^1$ that corresponds to the rational point $(a,b)$ such that $P_{a,b}\cap F^{\mathcal{G}}$ is splitting completely in $F/F^{\mathcal{G}}$ and
    \begin{equation}\label{fix}
        (\bigcup\limits_{\sigma\in\mathcal{G}}{\rm supp}((\sigma(x)-a)_0^{\mathbb{F}_q(x)}))\cap(\bigcup\limits_{\sigma\in\mathcal{G}}{\rm supp}((\sigma(x)-a)_\infty^{\mathbb{F}_q(x)}))=\emptyset.
    \end{equation}
    Let $\sigma_{2i-1}$ for $1\le i\le s$ be the automorphisms of $\mathcal{G}$ with the pairwise distinct $\sigma_{2i-1}(x)$ and $\sigma_1=1$. Let $\sigma_{2i}$ be the automorphisms of $\mathcal{G}$ with $\sigma_{2i-1}(x)=\sigma_{2i}(x)$ for $1\le i\le s$. Let $P_j=\sigma_j(P_{a,b})$ for $1\le j\le r+1$. Then there exists an element $z\in F$ satisfying that $F^{\mathcal{G}}=\mathbb{F}_q(z)$ and there exist elements $e_1,e_2,\cdots,e_r\in F$ that are linearly independent over $F^{\mathcal{G}}$ satisfying that
    \[
        \mathcal{L}(P+\sum_{u=1}^rP_u)=\langle\,e_1,\,\cdots,\,e_r\,\rangle
    \]
    where $P\in\mathbb{P}_F^1$ such that $P\notin \{P_j:1\le j\le r+1\}$.
\end{proposition}
\begin{proof}
    Let $z=\prod_{i=1}^s\frac{1}{\sigma_{2i-1}(x)-a}$. Then we have
    \[
        \sigma(z)=\prod_{i=1}^s\frac{1}{(\sigma\circ\sigma_{2i-1})(x)-a}=\prod_{i=1}^s\frac{1}{\sigma_{2i-1}(x)-a}=z
    \]
    for any $\sigma\in\mathcal{G}$ since $\sigma\circ\sigma_{2i-1}$ are pairwise distinct. It follows that $z\in F^{\mathcal{G}}$. The principal divisor of $z$ in $\mathbb{F}_q(x)$ is
    \[
        (z)^{\mathbb{F}_q(x)}=-\sum_{i=1}^s(\sigma_{2i-1}(x)-a)^{\mathbb{F}_q(x)}.
    \]
    Noting that, by (\ref{fix}) and \cite[Prop. 3.1.9]{Algebraic_Function_Fields_and_Codes}, we have $\deg\, (z)_\infty^{\mathbb{F}_q(x)}=s$ and 
    \begin{equation}\label{(z)}
        (z)^{F}={\rm Con}_{F/\mathbb{F}_q(x)}((z)^{\mathbb{F}_q(x)})={\rm Con}_{F/\mathbb{F}_q(x)}((z)_0^{\mathbb{F}_q(x)})-{\rm Con}_{F/\mathbb{F}_q(x)}((z)_\infty^{\mathbb{F}_q(x)})=(z)_0^F-\sum_{j=1}^{r+1}P_j.
    \end{equation}
    Thus $\deg\, (z)^F_\infty=r+1=[F:F^{\mathcal{G}}]$, so we have $F^{\mathcal{G}}=\mathbb{F}_q(z)$ by \cite[Thm. 1.4.11]{Algebraic_Function_Fields_and_Codes}. 

    Choose a rational place $P\in\mathbb{P}_F^1$ such that $P\notin \{P_j:1\le j\le r+1\}$. Note that 
    \begin{equation}\label{e_2}
        (\frac{1}{\sigma_1(x)-a})^F=(\frac{1}{x-a})^F=D_\infty-P_1-P_2.
    \end{equation}
    Therefore, $1\neq\frac{1}{x-a}\in\mathcal{L}(P_1+P_2)$. We set $e_1=1,\,e_2=\frac{1}{x-a}$. For each $3\le i\le r<q$, by the Riemann-Roch Theorem\cite[Thm. 1.5.15]{Algebraic_Function_Fields_and_Codes}, we have
    \[
        \#\bigcup\limits_{j=1}^i\mathcal{L}(P+\sum_{u=1}^iP_u-P_j)=iq^{i-1}<\#\mathcal{L}(P+\sum_{u=1}^iP_u)=q^i.
    \]
    It implies that there exists an element $e_i\in \mathcal{L}(P+\sum_{u=1}^iP_u)\setminus\mathcal{L}(P+\sum_{u=1}^iP_u-P_j)$ such that
    \begin{equation}\label{e_>=3}
        \begin{cases}
            v_{P^\prime}(e_i)\ge0,  &\forall P^\prime\notin\{P,P_1,\cdots,P_i\},\\
            v_{P_j}(e_i)=-1,  &\forall 1\le j\le i,\\
            v_{P}(e_i)\ge-1,
        \end{cases}
    \end{equation}
    for each $3\le i\le r$. Now we prove that $e_1,\cdots,e_r$ are linearly independent over $F^{\mathcal{G}}$. Suppose, on the contrary, that there is a non-trivial linear combination $\sum_{i=1}^r\varphi_i(z)e_i=0$ with $\varphi_i(z)\in\mathbb{F}_q(z)$. We can assume that all $\varphi_i(z)$ are polynomials in $z$. Let $1\le w\le r$ be the largest integer such that 
    \[
        \deg(\varphi_w(z))=\mathop{\max}_{1\le i\le w-1}\{\,\deg(\varphi_i(z))\,\}\quad{\rm and}\quad \deg(\varphi_w(z))>\mathop{\max}_{w+1\le i\le r}\{\,\deg(\varphi_i(z))\,\}.
    \]

    \begin{itemize}
        \item If $w>1$, then for $i<w$ we have
    \begin{align*}
        v_{P_{w}}(\varphi_i(z)e_i)&=-\deg(\varphi_i(z))+v_{P_{w}}(e_i)\\
        &\ge-\deg(\varphi_i(z))+0\\
        &>-\deg(\varphi_w(z))-1\\
        &=v_{P_{w}}(\varphi_w(z)e_w).
    \end{align*}
    For $i>w$, we have
    \begin{align*}
        v_{P_{w}}(\varphi_i(z)e_i)&=-\deg(\varphi_i(z))+v_{P_{w}}(e_i)\\
        &>-\deg(\varphi_w(z))+v_{P_{w}}(e_w)\\
        &=v_{P_{w}}(\varphi_w(z)e_w).
    \end{align*}
    The Strict Triangle Inequality \cite[Lem. 1.1.11]{Algebraic_Function_Fields_and_Codes} yields
    \[
        v_{P_{w}}(\sum_{i=1}^r\varphi_i(z)e_i)=v_{P_{w}}(\varphi_w(z)e_w)\neq v_{P_{w}}(0),
    \]
    which is a contradiction since $\sum_{i=1}^r\varphi_i(z)e_i=0$. 

        \item If $w=1$, then for $i>w$, we have
        \begin{align*}
        v_{P_{r+1}}(\varphi_i(z)e_i)&=-\deg(\varphi_i(z))+v_{P_{r+1}}(e_i)\\
        &>-\deg(\varphi_w(z))+v_{P_{r+1}}(e_w)\\
        &=v_{P_{r+1}}(\varphi_w(z)e_w).
        \end{align*}
        It follows that
    \[
        v_{P_{r+1}}(\sum_{i=1}^r\varphi_i(z)e_i)=v_{P_{r+1}}(\varphi_w(z)e_w)\neq v_{P_{r+1}}(0),
    \]
    which is also a contradiction. 
    \end{itemize}
    Hence, the elements $e_1,\cdots,e_r$ are linearly independent over $F^{\mathcal{G}}$ and satisfy
    \[
        \mathcal{L}(P+\sum_{u=1}^rP_u)=\langle\,e_1,\,\cdots,\,e_r\,\rangle.
    \]
    The proof is completed. 
\end{proof}

Moreover, if we choose the set $\mathcal{P}$ carefully, then we can guarantee that the assumption of Lemma \ref{local} is satisfied.

\begin{proposition}\label{e_i_local}
    Keep the notations in Proposition \ref{find_e_i}. Suppose the set $\mathcal{P}$ satisfies that
    \begin{equation}\label{choose_gen}
        \mathcal{P}\cap({\rm supp}((z)_{\infty}^F)\cup \{\bar{P}\})=\emptyset.
    \end{equation}
    Then every $r\times r$ submatrix of the matrix 
    \[
        M=\begin{pmatrix}
        e_1(P_{i,1}) & e_2(P_{i,1}) & \cdots & e_r(P_{i,1}) \\
        e_1(P_{i,2}) & e_2(P_{i,2}) & \cdots & e_r(P_{i,2}) \\
        \vdots & \vdots & \ddots & \vdots \\
        e_1(P_{i,r+1}) & e_2(P_{i,r+1}) & \cdots & e_r(P_{i,r+1})
        \end{pmatrix}
    \]
    is invertible for all $1\le i\le \ell$ and also every $r\times r$ submatrix of the matrix 
    \[
        M_z=\begin{pmatrix}
        e_1(P_{1}) & (e_2/z)(P_{1}) & \cdots & (e_r/z)(P_{1}) \\
        e_1(P_{2}) & (e_2/z)(P_{2}) & \cdots & (e_r/z)(P_{2}) \\
        \vdots & \vdots & \ddots & \vdots \\
        e_1(P_{r+1}) & (e_2/z)(P_{r+1}) & \cdots & (e_r/z)(P_{r+1})
        \end{pmatrix}
    \]
    is invertible.
\end{proposition}
\begin{proof}
    For $1\le i\le \ell$, let $P_{i,1},P_{i,2},\cdots,P_{i,r+1}$ be the $r+1$ rational places lying over $Q_i$ that corresponds to the zero of $z-a_i$ for some $a_i\in\mathbb{F}_q$. Without loss of generality, we can consider the first $r$ rows of the matrix $M$. Suppose that there exists ${\bf 0}\neq(c_1,\cdots,c_r)\in\mathbb{F}_q^r$ such that
    \[
        \begin{pmatrix}
        e_1(P_{i,1}) & e_2(P_{i,1}) & \cdots & e_r(P_{i,1}) \\
        e_1(P_{i,2}) & e_2(P_{i,2}) & \cdots & e_r(P_{i,2}) \\
        \vdots & \vdots & \ddots & \vdots \\
        e_1(P_{i,r}) & e_2(P_{i,r}) & \cdots & e_r(P_{i,r})
        \end{pmatrix}
        \begin{pmatrix}
        c_1 \\
        c_2 \\
        \vdots \\
        c_r
        \end{pmatrix}=0.
    \]
    Set $f=c_1e_1+\cdots+c_re_r$. Then for $1\le j \le r$, we have 
    \[
        v_{P_{i,j}}(f)\ge1\quad {\rm i.e.}\quad \deg\, (f)^F_0\ge r.
    \]
    On the other hand, $f\in\mathcal{L}(P+\sum_{j=1}^rP_j)$ i.e. $\deg(f)_{\infty}\le r+1$. Thus there exists a rational place $P^\prime\in \mathbb{P}_F^1$ such that
    \begin{equation}\label{(f)}
        (f)^F=P^\prime+\sum_{j=1}^rP_{i,j}-(P+\sum_{j=1}^rP_j).
    \end{equation}
    Note that $(z)_\infty^{\mathbb{F}_q(z)}=(z-a_i)_\infty^{\mathbb{F}_q(z)}$, so by (\ref{(z)}) we have 
    \begin{equation}\label{(z-a_i)}
        (z-a_i)^F=\sum_{j=1}^{r+1}P_{i,j}-\sum_{j=1}^{r+1}P_j.
    \end{equation}
    By (\ref{(f)}) and (\ref{(z-a_i)}), we conclude
    \[
        (f/(z-a_i))^F=P^\prime+P_{r+1}-(P+P_{i,r+1}),
    \]
    i.e.
    \[
        P^\prime+P_{r+1}-D_\infty=P+P_{i,r+1}-D_\infty+(f/(z-a_i))^F.
    \]
    It follows that 
    \[
        [P^\prime+P_{r+1}-D_\infty]=[P+P_{i,r+1}-D_\infty]\in\text{Cl}^0(F),
    \]
    which is a contradiction by Lemmas \ref{unique} and (\ref{choose_gen}).

    For the matrix $M_z$, without loss of generality, we can also consider the first $r$ rows. Suppose that there exists ${\bf 0}\neq(c^\prime_1,\cdots,c^\prime_r)\in\mathbb{F}_q^r$ such that
    \[
        \begin{pmatrix}
        e_1(P_{1}) & (e_2/z)(P_{1}) & \cdots & (e_r/z)(P_{1}) \\
        e_1(P_{2}) & (e_2/z)(P_{2}) & \cdots & (e_r/z)(P_{2}) \\
        \vdots & \vdots & \ddots & \vdots \\
        e_1(P_{r}) & (e_2/z)(P_{r}) & \cdots & (e_r/z)(P_{r})
        \end{pmatrix}
        \begin{pmatrix}
        c^\prime_1 \\
        c^\prime_2 \\
        \vdots \\
        c^\prime_r
        \end{pmatrix}=0.
    \]
    Set $f_z=c^\prime_1e_1+c^\prime_2(e_2/z)+\cdots+c^\prime_r(e_r/z)$. Then there exist rational places $\hat{P},\hat{P}^\prime\in\mathbb{P}_F^1$ such that
    \[
        (f_z)^F=\hat{P}+\hat{P}^\prime+\sum_{j=1}^rP_j-(P+(z)_0^F)=\hat{P}+\hat{P}^\prime-(P+P_{r+1})+(1/z)^F,
    \]
    i.e.
    \[
        \hat{P}+\hat{P}^\prime-D_\infty=P+P_{r+1}-D_\infty+(zf_z)^F,
    \]
    since $f_z\in\mathcal{L}(P+(z)_0^F)$ and (\ref{(z)}). By the proof of Proposition \ref{find_e_i}, we have $\bar{P}_{r+1}=P_r$. Therefore there must be $\hat{P}=P_{r+1}$ or $\hat{P}^\prime=P_{r+1}$ by Lemma \ref{unique} which means that
    \[
        f_z(P_{r+1})=0.
    \]
    It follows that
    \begin{equation}\label{M_zc_i}
        M_z\cdot\begin{pmatrix}
        c^\prime_1 \\
        c^\prime_2 \\
        \vdots \\
        c^\prime_r
        \end{pmatrix}=0.
    \end{equation}
    By (\ref{e_2}) and (\ref{e_>=3}), we have
    \begin{equation}\label{M_z_specific}
        M_z=
        \begin{pmatrix}
        1 & \star & \star & \cdots & \star & \star & \star \\
        1 & \star & \star & \cdots & \star & \star & \star \\
        1 & 0 & \star & \cdots & \star & \star & \star \\
        1 & 0 & 0 & \cdots & \star & \star & \star \\
        \vdots & \vdots & \vdots & \ddots & \vdots & \vdots & \vdots \\
        1 & 0 & 0 & \cdots & 0 & \star & \star \\
        1 & 0 & 0 & \cdots & 0 & 0 & \star \\
        1 & 0 & 0 & \cdots & 0 & 0 & 0 \\
        \end{pmatrix},
    \end{equation}
    where $\star$ represents an element in $\mathbb{F}_q\setminus\{0\}$. Hence,  for any $1\le i\le r$, we obtain $c^\prime_i=0$ from (\ref{M_zc_i}) and (\ref{M_z_specific}) which contradicts $(c^\prime_1,\cdots,c^\prime_r)\neq{\bf 0}$. The proof is completed. 
\end{proof}

With the propositions above, we can give a general framework for constructing locally repairable codes via automorphism groups of hyperelliptic curves with genus $2$. 

\begin{proposition}\label{con_1}
    Let $\mathfrak{C}/\mathbb{F}_q$ be a hyperelliptic curve with function field $F$ defined by the hyperelliptic model (\ref{hyper_eqa}). Let $\mathcal{G}$ be a subgroup of ${\rm{Aut}}(\mathfrak{C}/\mathbb{F}_q)$ and satisfy the assumptions of Proposition \ref{find_e_i}. Let $\mathcal{P}$ be defined as (\ref{ev_P}) and satisfy the assumption of Proposition \ref{e_i_local}. Denote $N=\#\mathbb{P}_F^1$. Then there exists an either optimal or almost optimal q-ary
    \[
        [n=\ell(r+1),\,k=rt-(r-1),\,d\ge n-(t-1)(r+1)-1]_q
    \]
    locally repairable code with locality $r$ for any $1\le \ell\le\lfloor\frac{N-2r-4}{r+1}\rfloor$ and $1\le t\le \ell$.
\end{proposition}
\begin{proof}
    Let $z\in F$ and $e_1,e_2,\cdots,e_r\in F$ be elements in Proposition \ref{find_e_i}. By (\ref{Hurwitz_genus_formula}) and $F^{\mathcal{G}}=\mathbb{F}_q(z)$, we have
    \[
        2=2g(F)-2=[F:F^{\mathcal{G}}](2g(F^{\mathcal{G}})-2)+\deg {\rm{Diff}}(F/F^{\mathcal{G}})=-2(r+1)+\deg {\rm{Diff}}(F/F^{\mathcal{G}}).
    \]
    It follows that
    \[
        \deg {\rm{Diff}}(F/F^{\mathcal{G}})=2r+4,
    \]
    which means that there are at most $2r+4$ rational places in total that are ramified in $F/F^{\mathcal{G}}$. Set $f_j=z^{j-1}$ for $1\le j\le t$ and consider the set $V$ defined by (\ref{func}). Let $G=(t-1)(P_1+\cdots+P_{r+1})+P$ be a divisor of $F$ where $P,P_1,\cdots,P_{r+1}$ are given in Proposition \ref{find_e_i}. Indeed, regarding the choice of $P$, we can choose a rational place that is ramified in $F/F^{\mathcal{G}}$, at which point the assumption of Proposition \ref{e_i_local} will be easily satisfied. Then we can verify that $V$ is a subspace of $\mathcal{L}(G)$. By Lemma \ref{cons}, Proposition \ref{e_i_local} and $\mathcal{P}$, the algebraic geometry code $C_{\mathcal{L}}(\mathcal{P},V)$ is a $q$-ary 
    \[
        [n=\ell(r+1),\,k=rt-(r-1),\,d\ge n-(t-1)(r+1)-1]_q
    \]
    locally repairable code with locality $r$ for $1\le \ell\le\lfloor\frac{N-3r-5}{r+1}\rfloor$ and $1\le t\le \ell$. 
    
    Let $\tilde{\mathcal{P}}=\mathcal{P}\cup\{P_1,\cdots,P_{r+1}\}$ and 
    \[
        \pi_R=\begin{cases}
            1, &\forall R\in\mathcal{P},\\
            (1/z)^{t-1}, &\forall R\in\{P_1,\cdots,P_{r+1}\}.
        \end{cases}
    \]
    Then the modified algebraic geometry code $C_{\mathcal{L}}(\tilde{\mathcal{P}},V)$ is still a 
    \[
        [n=\ell(r+1),\,k=rt-(r-1),\,d\ge n-(t-1)(r+1)-1]_q
    \]
    linear code for $1\le \ell\le\lfloor\frac{N-2r-4}{r+1}\rfloor$. Let 
    \[
        f=\sum_{j=1}^ta_{1,j}f_je_1+\sum_{i=2}^r\big(\sum_{j=1}^{t-1}a_{i,j}f_j\big)e_i\in V,
    \]
    for some $a_{i,j}\in\mathbb{F}_q$. We claim that the value of $(1/z)^{t-1}f\in V$ at any place in the set $\{P_1,\cdots,P_{r+1}\}$ can be recovered from the values of $(1/z)^{t-1}f$ at the other $r$ places of this set. For every $1\le m\le r+1$, we have
    \begin{align*}
        (1/z)^{t-1}f(P_m)&=\sum_{j=1}^ta_{1,j}(f_j/z^{t-1})(P_m)+\sum_{i=2}^r\big(\sum_{j=1}^{t-1}a_{i,j}(f_j/z^{t-2})(P_m)\big)(e_i/z)(P_m) \\
        &=\sum_{j=1}^ta_{1,j}(z^{j-1}/z^{t-1})(P_m)+\sum_{i=2}^r\big(\sum_{j=1}^{t-1}a_{i,j}(z^{j-1}/z^{t-2})(P_m)\big)(e_i/z)(P_m) \\
        &=a_{1,t}+\sum_{i=2}^ra_{i,t-1}(e_i/z)(P_m).
    \end{align*}
    Note that every $r\times r$ submatrix of the matrix $M_z$ is invertible by Proposition \ref{e_i_local}. Since we know the value of $(e_i/z)(P_m)$ for each $i$ and $m$, we can calculate the value $a_{1,t}$ and the values of $a_{i,t-1}$ from the values of $(1/z)^{t-1}f$ at the other $r$ places of the set $\{P_1,\cdots,P_{r+1}\}$. Thus the claim is proved. It means that $C_{\mathcal{L}}(\tilde{\mathcal{P}},V)$ is still a locally repairable code with locality $r$. By the bound (\ref{min_d_bound}), 
    \begin{align*}
        d&\le n-k-\bigg\lceil\frac{k}{r}\bigg\rceil+2\\
        &=n-rt+r-1-\bigg\lceil\frac{rt-r+1}{r}\bigg\rceil+2\\
        &=n-(t-1)(r+1),
    \end{align*}
    which implies that $C_{\mathcal{L}}(\tilde{\mathcal{P}},V)$ is an either optimal or almost optimal locally repairable code. The proof is completed. 
\end{proof}

\begin{remark}\label{upper_bound}
    The upper bound of $\ell$ in the above Proposition is just a rough estimate. In most cases, the number of the rational places which are ramified in $F/F^{\mathcal{G}}$ is relatively small. Therefore, the maximal length $n$ of the code that we constructed will be very close to $N$. 
\end{remark}

\subsection{Locality r odd}

Considering the subgroup of ${\rm{Aut}}(\mathfrak{C}/\mathbb{F}_q)$ that aligns with the conditions specified in Proposition \ref{find_e_i}, it follows that the locality $r$ must be an odd integer. Utilizing the lemmas presented in Section \ref{pre}, we proceed to construct some certain locally repairable codes. These codes are derived from hyperelliptic curves of genus 2, with lengths that can approach $q + 4\sqrt{q}$.

\begin{theorem}\label{r_odd}
    Let $q=p^{2a}$ be a power of an odd prime $p$ and an odd integer $a>0$. Then there exists an either optimal or almost optimal $q$-ary
    \[
        [n=\ell(r+1),\,k=rt-(r-1),\,d\ge n-(t-1)(r+1)-1]_q
    \]
    locally repairable code with locality $r$ for any $1\le t\le\ell\le\big\lfloor\frac{q+4\sqrt{q}-2r-3}{r+1}\big\rfloor$ if $p$ and $r$ satisfy one of the following conditions:
    \begin{itemize}
        \item[(1)] $r=3,7,\;p\equiv 5$ or $7 \pmod{8}$.
        \item[(2)] $r=9,\;p\equiv -1 \pmod{5}$.
    \end{itemize}
    
    Moreover, suppose that the subgroup $\mathcal{G}\subseteq{\rm{Aut}}(\mathfrak{C}/\mathbb{F}_q)$ satisfies that $-I\in\mathcal{G}$ and
    \begin{equation}\label{degree}
        \text{the degree of the numerator of } \prod_{\sigma\in\mathcal{G}}(\sigma(x)-a)=\sum_{\sigma\in\mathcal{G}}\text{the degree of the numerator of }\sigma(x)-a.
    \end{equation}
    Then $r$ can be chosen more flexibly as follows ($r$ is odd):
    \begin{itemize}
        \item[(3)] $r+1\mid 48,\;p\equiv 5$ or $7 \pmod{8}$ and $p\neq5$.
        \item[(4)] $r+1\mid 240,\;p=5$.
    \end{itemize}
\end{theorem}
\begin{proof}
    (1) We consider the hyperelliptic curve $\mathfrak{C}:y^2=x^5+x$. By Lemmas \ref{max_hy}(i) and \ref{maximality}, we know that the hyperelliptic curve $\mathfrak{C}$ is maximal over $\mathbb{F}_q$. Moreover, we can verify that $8\mid q-1$ and $2^{1/2}\in\mathbb{F}_q$ since $p\equiv 5$ or $7 \pmod{8}$ and $2\mid [\mathbb{F}_q:\mathbb{F}_p]$. Let $F$ be the function field $\mathbb{F}_q(\mathfrak{C})$.
    
    If $p\neq5$, then we have ${\rm{Aut}}(\mathfrak{C}/\mathbb{F}_q)\simeq \tilde{S}_4$ by Lemma \ref{auto_max_1}(i). We consider $\mathcal{G}=\langle V\rangle$, where 
    \[
        V=2^{-1/2}\begin{pmatrix}
        (-1)^{1/2}-1 & 0 \\
        0 & (-1)^{1/2}+1
        \end{pmatrix}.
    \]
    It follows that $\sigma(x)-a$ in Proposition \ref{find_e_i} are all polynomials, which means that the poles of $\sigma(x)-a$ must not be the zeros of $\sigma^\prime(x)-a$ where $\sigma\neq\sigma^\prime\in\mathcal{G}$. By (\ref{trans_eq}), it is easy to obtain $P_\infty\cap F^{\mathcal{G}}$ is ramified in $F/F^{\mathcal{G}}$. Hence, we can choose $P=P_\infty$ in Proposition \ref{find_e_i}. Thus, by Proposition \ref{con_1} there exists an either optimal or almost optimal $q$-ary locally repairable code with locality $r=3,7$.

    If $p=5$, then ${\rm{Aut}}(\mathfrak{C}/\mathbb{F}_q)\simeq \tilde{S}_5$ by Lemma \ref{auto_max_1}(ii). Similar to the above proof, we consider $\mathcal{G}=\langle W\rangle$, where 
    \[
        W=2^{1/2}\begin{pmatrix}
        1 & 0 \\
        0 & 2
        \end{pmatrix}.
    \]
    Then, by Proposition \ref{con_1} there exists an either optimal or almost optimal $q$-ary locally repairable code with locality $r=3,7$.
    
    (2) We now consider the hyperelliptic curve $\mathfrak{C}:y^2=x^5+1$. By Lemmas \ref{max_hy}(ii) and \ref{maximality}, we can know that the hyperelliptic curve $\mathfrak{C}$ is maximal over $\mathbb{F}_q$. In addition, we can verify that $5\mid q-1$ since $p\equiv -1 \pmod{5}$. Therefore, ${\rm{Aut}}(\mathfrak{C}/\mathbb{F}_q)=\langle\,U,\,V\,\rangle\simeq C_{10}$ by Lemma \ref{auto_max_2} where
    \[
        U=\begin{pmatrix}
        -1 & 0 \\
        0 & -1
        \end{pmatrix}\quad {\rm and} \quad V=\begin{pmatrix}
        \zeta^2 & 0 \\
        0 & \zeta
        \end{pmatrix}.
    \]
    We consider $\mathcal{G}={\rm{Aut}}(\mathfrak{C}/\mathbb{F}_q)$. The rest of the proof is completely similar to (1).

    In the following, we prove (3) and (4). Note that the equation (\ref{degree}) implies that, for any two distinct $\sigma,\sigma^\prime\in\mathcal{G}$, the poles of $\sigma(x)-a$ must not be the zeros of $\sigma^\prime(x)-a$. It follows that the assumptions in Proposition \ref{find_e_i} are satisfied. Therefore, for (3) we could consider the subgroup $\mathcal{G}\subseteq{\rm{Aut}}(\mathfrak{C}/\mathbb{F}_q)\simeq \tilde{S}_4$ where $|\mathcal{G}|$ is even, for (4) we could consider the subgroup $\mathcal{G}\subseteq{\rm{Aut}}(\mathfrak{C}/\mathbb{F}_q)\simeq \tilde{S}_5$ where $|\mathcal{G}|$ is even. Now the proof is completed. 
\end{proof}

\begin{remark}\label{r_odd_re}
    {\rm (1)} It should be noted that the construction above also applies to the case that $p^a\equiv 5$ or $7 \pmod{8}$ or $p^a\equiv -1 \pmod{5}$ where $a$ can be any positive integer. 
    
    {\rm (2)} We can consider the hyperelliptic curve $\mathfrak{C}$ over $\mathbb{F}_q$ in Lemma \ref{auto_non}. Then, similar to the proof of Theorem \ref{r_odd} above, we can also obtain some either optimal or almost optimal $q$-ary locally repairable codes with locality $r=3,5,7$ or $11$. At this point, there are almost no restrictions on the prime $p$. 
\end{remark}

\begin{example}\label{r_odd_almost}
    {\rm (1)} Let $q=13^2$ and the hyperelliptic curve $\mathfrak{C}:y^2=x^5+x$. Then, we have ${\rm{Aut}}(\mathfrak{C}/\mathbb{F}_q)\simeq \tilde{S}_4$. More specifically,  
    \[
        {\rm{Aut}}(\mathfrak{C}/\mathbb{F}_q)=\langle\,U,\,V\mid U^2=(UV)^3=I;\,V^4=-I=(UVU)^4\,\rangle\simeq \tilde{S}_4, 
    \]
    where 
    \[
        U=2^{-1/2}\begin{pmatrix}
        1 & -(-1)^{1/4} \\
        (-1)^{3/4} & -1
        \end{pmatrix}\quad {\it and} \quad V=2^{-1/2}\begin{pmatrix}
        (-1)^{1/2}-1 & 0 \\
        0 & (-1)^{1/2}+1
        \end{pmatrix}.
    \]
    
    Let $\mathbb{F}_q=\mathbb{F}_{13}(u)$, by calculation, we can list all the transformations of automorphisms $\sigma\in{\rm{Aut}}(\mathfrak{C}/\mathbb{F}_q)$ on $x$ in Table \ref{tab_S_4}.
    \begin{table}[ht]
	\centering
	\begin{tabular}{cccccc}
		\hline
		$\sigma$ & $\sigma(x)$ & $\sigma$ & $\sigma(x)$ & $\sigma$ & $\sigma(x)$  \\ \hline
		$U$ & $\frac{(3u+5)x-2}{3x-(3u+5)}$ & $(UV^2)^2$ & $\frac{u+6}{(-5u-4)x}$ & $UVUV^3$ & $\frac{-3x+(3u+5)}{-(3u+5)x+2}$ \\ 
  
		$U^2$ & $x$ & $(UV^2)^3$ & $\frac{(2u-1)x+3}{-2x+(2u-1)}$ & $(UVUV^3)^2$ & $\frac{2x-(3u+5)}{(3u+5)x-3}$ \\ 
  
		$V$ & $-5x$ & $UV^3$ & $\frac{-2x-(3u+5)}{(3u+5)x+3}$ & $UV^2UV$ & $\frac{1}{x}$ \\ 
  
		$V^2$ & $-x$ & $(UV^3)^2$ & $\frac{-3x-(3u+5)}{(3u+5)x+2}$ & $UV^2UV^3$ & $\frac{-1}{x}$ \\ 
  
		$V^3$ & $5x$ & $UVU$ & $\frac{(-2u+1)x+2}{-3x+(-2u+1)}$ & $UV^3UV$ & $\frac{2x+(-2u+1)}{(-2u+1)x-3}$ \\ 
  
		$UV$ & $\frac{-3x+(-2u+1)}{(-2u+1)x+2}$ & $(UVU)^2$ & $\frac{5u+4}{(-u-6)x}$ & $(UV^3UV)^2$ & $\frac{-3x+(2u-1)}{(2u-1)x+2}$ \\ 
  
		$(UV)^2$ & $\frac{2x+(2u-1)}{(2u-1)x-3}$ & $(UVU)^3$ & $\frac{(-2u+1)x-2}{3x+(-2u+1)}$ & $UV^3UV^2$ & $\frac{(3u+5)x-3}{2x-(3u+5)}$ \\ 
  
		$UV^2$ & $\frac{(2u-1)x-3}{2x+(2u-1)}$ & $UVUV^2$ & $\frac{(3u+5)x+3}{-2x-(3u+5)}$ & $UVUV^2UV^3$ & $\frac{(3u+5)x+2}{-3x-(3u+5)}$ \\
        \hline
	\end{tabular}
	\caption{Automorphisms $\sigma\in{\rm{Aut}}(\mathfrak{C}/\mathbb{F}_q)$ on $x$}
	\label{tab_S_4}
    \end{table}
    
    Choose $P_{a,b}=P_{2,u-7}$ and $P=P_\infty$ in Proposition \ref{find_e_i}. By Sagemath, we can verify that
    \begin{align*}
        \text{the degree of the numerator of } &\prod_{\sigma\in{\rm{Aut}}(\mathfrak{C}/\mathbb{F}_q)}(\sigma(x)-2)\\
        &=\sum_{\sigma\in{\rm{Aut}}(\mathfrak{C}/\mathbb{F}_q)}\text{the degree of the numerator of }\sigma(x)-2.
    \end{align*}
    Therefore, by Theorem \ref{r_odd}(3) for any subgroup $\mathcal{G}\subseteq{\rm{Aut}}(\mathfrak{C}/\mathbb{F}_q)$ satisfying $-I\in\mathcal{G}$, we can obtain some either optimal or almost optimal $13^2$-ary locally repairable codes with locality $|\mathcal{G}|-1$ in Table \ref{LRC_y2=x5+x_1}.
    \begin{table}[ht]
	\centering
	\begin{tabular}{ccc}
        \hline
		Locality $r$ & Our locally repairable codes & The range of $\ell$ \\
        \hline
		$3$ & $[n=4\ell,\,k=3t-2,\,d\ge n-4(t-1)-1]_{13^2}$ & $1\le t\le\ell\le54$ \\
        
		$5$ & $[n=6\ell,\,k=5t-4,\,d\ge n-6(t-1)-1]_{13^2}$ & $1\le t\le\ell\le36$ \\
        
		$7$ & $[n=8\ell,\,k=7t-6,\,d\ge n-8(t-1)-1]_{13^2}$ & $1\le t\le\ell\le27$ \\
        
		$11$ & $[n=12\ell,\,k=11t-10,\,d\ge n-12(t-1)-1]_{13^2}$ & $1\le t\le\ell\le17$ \\
        
		$15$ & $[n=16\ell,\,k=15t-14,\,d\ge n-16(t-1)-1]_{13^2}$ & $1\le t\le\ell\le13$ \\
        
		$23$ & $[n=24\ell,\,k=23t-22,\,d\ge n-24(t-1)-1]_{13^2}$ & $1\le t\le\ell\le9$ \\
        
		$47$ & $[n=48\ell,\,k=47t-46,\,d\ge n-48(t-1)-1]_{13^2}$ & $1\le t\le\ell\le4$ \\
        \hline
	\end{tabular}
	\caption{Locally repairable codes over $\mathbb{F}_{13^2}$}
	\label{LRC_y2=x5+x_1}
    \end{table}
    
    {\rm (2)} Let $q=5^6$ and the hyperelliptic curve $\mathfrak{C}:y^2=x^5+x$. Then we have ${\rm{Aut}}(\mathfrak{C}/\mathbb{F}_q)\simeq \tilde{S}_5$. Let $\mathbb{F}_q=\mathbb{F}_{5}(u)$. Choose $P_{a,b}=P_{u,3u^5+2u^4+u^3+u^2+2}$ and $P=P_\infty$ in Proposition \ref{find_e_i}. By Sagemath, we can verify that
    \begin{align*}
        \text{the degree of the numerator of } &\prod_{\sigma\in{\rm{Aut}}(\mathfrak{C}/\mathbb{F}_q)}(\sigma(x)-u)\\
        &=\sum_{\sigma\in{\rm{Aut}}(\mathfrak{C}/\mathbb{F}_q)}\text{the degree of the numerator of }\sigma(x)-u.
    \end{align*}
    Therefore, by Theorem \ref{r_odd}(4) for any subgroup $\mathcal{G}\subseteq{\rm{Aut}}(\mathfrak{C}/\mathbb{F}_q)$ satisfying $-I\in\mathcal{G}$, we can obtain some either optimal or almost optimal $5^6$-ary locally repairable codes with locality $|\mathcal{G}|-1$ in Table \ref{LRC_y2=x5+x_2}.
    \begin{table}[ht]
	\centering
	\begin{tabular}{ccc}
        \hline
		Locality $r$ & Our locally repairable codes & The range of $\ell$ \\
        \hline
		$3$ & $[n=4\ell,\,k=3t-2,\,d\ge n-4(t-1)-1]_{5^6}$ & $1\le t\le\ell\le4030$ \\
        
		$5$ & $[n=6\ell,\,k=5t-4,\,d\ge n-6(t-1)-1]_{5^6}$ & $1\le t\le\ell\le2686$ \\
        
		$7$ & $[n=8\ell,\,k=7t-6,\,d\ge n-8(t-1)-1]_{5^6}$ & $1\le t\le\ell\le2014$ \\
        
        $9$ & $[n=10\ell,\,k=9t-0,\,d\ge n-10(t-1)-1]_{5^6}$ & $1\le t\le\ell\le1612$ \\
        
		$11$ & $[n=12\ell,\,k=11t-10,\,d\ge n-12(t-1)-1]_{5^6}$ & $1\le t\le\ell\le1343$ \\
       
		$15$ & $[n=16\ell,\,k=15t-14,\,d\ge n-16(t-1)-1]_{5^6}$ & $1\le t\le\ell\le1007$ \\
        
		$19$ & $[n=20\ell,\,k=19t-18,\,d\ge n-20(t-1)-1]_{5^6}$ & $1\le t\le\ell\le806$ \\
       
		$23$ & $[n=24\ell,\,k=23t-22,\,d\ge n-24(t-1)-1]_{5^6}$ & $1\le t\le\ell\le671$ \\
        
        $39$ & $[n=40\ell,\,k=39t-38,\,d\ge n-40(t-1)-1]_{5^6}$ & $1\le t\le\ell\le403$ \\
        
        $47$ & $[n=48\ell,\,k=47t-46,\,d\ge n-48(t-1)-1]_{5^6}$ & $1\le t\le\ell\le335$ \\
        
        $119$ & $[n=120\ell,\,k=119t-118,\,d\ge n-120(t-1)-1]_{5^6}$ & $1\le t\le\ell\le134$ \\
        
        $239$ & $[n=240\ell,\,k=239t-238,\,d\ge n-240(t-1)-1]_{5^6}$ & $1\le t\le\ell\le67$ \\
        \hline
	\end{tabular}
	\caption{Locally repairable codes over $\mathbb{F}_{5^6}$}
	\label{LRC_y2=x5+x_2}
    \end{table}

    {\rm (3)} Let $q=9^2$ and the hyperelliptic curve $\mathfrak{C}:y^2=x^5+1$. Choose $P_{a,b}=P_{1,2^{1/2}}$ and $P=P_\infty$ in Proposition \ref{find_e_i}. Then, by Theorem \ref{r_odd}(2) and Remark \ref{r_odd_re}(1), we can obtain an either optimal or almost optimal $9^2$-ary
    \[
        [n=10\ell,\,k=9t-8,\,d\ge n-10(t-1)-1]_{9^2}
    \]
    locally repairable code with locality $9$ for any $1\le t\le\ell\le11$. 

    {\rm (4)} Let $q=11^2$. Note that $p=11$ does not satisfy the conditions in Theorem \ref{r_odd}, thus we can not use the hyperelliptic curve above. By Remark \ref{r_odd_re}(2), we could consider the curve $\mathfrak{C}:y^2=x^5+x^3+6x$. Since $4\mid q-1$ and $6^{1/4}\in\mathbb{F}_q$, then by Lemma \ref{auto_non}(i), we have ${\rm{Aut}}(\mathfrak{C}/\mathbb{F}_q)=\langle \,U,\,V\,\rangle\simeq D_8$ where 
    \[
        U=\begin{pmatrix}
        -(-1)^{1/2} & 0 \\
        0 & (-1)^{1/2}
        \end{pmatrix}\quad {\it and} \quad V=\begin{pmatrix}
        0 & 6^{1/4} \\
        6^{-1/4} & 0
        \end{pmatrix}.
    \]
    
    Choose $P_{a,b}=P_{1,8^{1/2}}$ and $P=P_\infty$ in Proposition \ref{find_e_i}. Then, it is easy to verify that
    \begin{align*}
        \text{the degree of the numerator of } &\prod_{\sigma\in{\rm{Aut}}(\mathfrak{C}/\mathbb{F}_q)}(\sigma(x)-1)\\
        &=\sum_{\sigma\in{\rm{Aut}}(\mathfrak{C}/\mathbb{F}_q)}\text{the degree of the numerator of }\sigma(x)-1.
    \end{align*}
    Therefore, for any subgroup $\mathcal{G}\subseteq{\rm{Aut}}(\mathfrak{C}/\mathbb{F}_q)$, we can obtain some either optimal or almost optimal $11^2$-ary locally repairable codes with locality $|\mathcal{G}|-1$ in Table \ref{LRC_y2=x5+x3+6x}.
    \begin{table}[ht]
	\centering
	\begin{tabular}{ccc}
        \hline
		Locality $r$ & Our locally repairable codes & The range of $\ell$ \\
        \hline
		$3$ & $[n=4\ell,\,k=3t-2,\,d\ge n-4(t-1)-1]_{11^2}$ & $1\le t\le\ell\le40$ \\

		$7$ & $[n=8\ell,\,k=7t-6,\,d\ge n-8(t-1)-1]_{11^2}$ & $1\le t\le\ell\le19$ \\
        \hline
	\end{tabular}
	\caption{Locally repairable codes over $\mathbb{F}_{11^2}$}
	\label{LRC_y2=x5+x3+6x}
    \end{table}
\end{example}

As we have seen here, all the upper bounds of $\ell$ in Example \ref{r_odd_almost} are greater than $\big\lfloor\frac{q+4\sqrt{q}-2r-3}{r+1}\big\rfloor$. This is consistent with Remark \ref{upper_bound}. In the following, we will give some specific examples about optimal locally repairable codes. 

\begin{example}\label{r_odd_optimal}
    {\rm{(1)}} Let $q=3^2$ and $\mathbb{F}_q=\mathbb{F}_3(u)$. Consider the hyperelliptic curve $\mathfrak{C}:y^2=x^5+x^3+2x$. By Lemma \ref{auto_non}(i), we can take 
    \[
        \mathcal{G}=\langle \,U=\begin{pmatrix}
        -(-1)^{1/2} & 0 \\
        0 & (-1)^{1/2}
        \end{pmatrix}\,\rangle.
    \]
    Then $\{\sigma(x):\sigma\in\mathcal{G}\}=\{x,\,-x\}$. Following Proposition \ref{find_e_i}, choose $P_{a,b}=P_{1,1}$ and $P=P_\infty$, then it is easy to see that
    \[
        \deg((x-1)(-x-1))=\deg(x-1)+\deg(-x-1).
    \]
    Hence, we can find $e_1,e_2,e_3$ such that
    \[
        \mathcal{L}(P_\infty+P_{1,1}+P_{1,2}+P_{2,u+1})=\langle\,e_1,e_2,e_3\,\rangle.
    \]
    More precisely, the elements $e_1=1,\,e_2=\frac{x}{x+2},\,e_3=\frac{y+u+1}{x^2+2},\,z=\frac{1}{x^2+2}$ and there are three other sets of evaluation points
    \[
        \{P_{u,1},P_{u,2},P_{2u,u+1},P_{2u,2u+2}\},
    \]
    \[
        \{P_{u+2,u+1},P_{u+2,2u+2},P_{2u+1,1},P_{2u+1,2}\},
    \]
    \[
        \text{and } \{P_{2u+2,u},P_{2u+2,2u},P_{u+1,u+2},P_{u+1,2u+1}\}.
    \]
    
    By Proposition \ref{con_1} and Sagemath, we can construct the following optimal $3^2$-ary locally repairable codes with locality $3$ in Table \ref{opt_LRC_y2=x5+x3+2x}.
    \begin{table}[ht]
	\centering
	\begin{tabular}{c}
        \hline
		Our optimal locally repairable codes \\
        
        $[n=4\ell,\,k=3t-2,\,d=n-4(t-1)]_{3^2}$\\
        
        \hline
		$[8,\,4,\,4]_{3^2}$ \\

		$[12,\,4,\,8]_{3^2}$ \\

        $[12,\,7,\,4]_{3^2}$ \\

        $[16,\,10,\,4]_{3^2}$ \\
        \hline
	\end{tabular}
	\caption{Optimal locally repairable codes with locality $3$ over $\mathbb{F}_{3^2}$}
	\label{opt_LRC_y2=x5+x3+2x}
    \end{table}

    {\rm(2)} Let $q=5^2$ and $\mathbb{F}_q=\mathbb{F}_5(u)$. Consider the hyperelliptic curve $\mathfrak{C}:y^2=x^5+x$. By Lemma \ref{auto_max_1}(ii), we take
    \[
        \mathcal{G}=\langle \,V^\prime=\begin{pmatrix}
        0 & -(-1)^{-1/4}\cdot2 \\
        -(-1)^{1/4}\cdot2 & 1
        \end{pmatrix}\,\rangle.
    \]
    Then $\{\sigma(x):\sigma\in\mathcal{G}\}=\{x,\,\frac{x+2u-1}{(u+2)x},\,\frac{2u-1}{(u+2)x-1}\}$. Choose $P_{a,b}=P_{1,u+2}$ and $P=P_\infty$ in Proposition \ref{find_e_i}, we obtain
    \[
        \prod_{\sigma\in\mathcal{G}}(\sigma(x)-1)=(x-1)\cdot\frac{(4u+4)x+(2u-1)}{(u+2)x}\cdot\frac{(4u+3)x+2u}{(u+2)x-1}=\frac{4ux^3+2ux^2+(u+2)x-(2u+2)}{2x^2-(u+2)x}.
    \]
    Thus
    \begin{align*}
        \text{the degree of the numerator of } &\prod_{\sigma\in\mathcal{G}}(\sigma(x)-1)\\
        &=\sum_{\sigma\in\mathcal{G}}\text{the degree of the numerator of }\sigma(x)-1.
    \end{align*}
    Hence, we can find the elements $e_1=1$, $e_2=\frac{x}{x+4}$, $e_3=\frac{y+2}{x^2+(2u+1)x+3u+3}$, $e_4=\frac{y}{x^2+(2u+1)x+3u+3}$, $e_5=\frac{(x+3u+3)y+2u+4}{x^3 + 3x^2 + (u + 3)x + 4u + 3}$, $z=\frac{1}{x^3 + 3x^2 + (u + 3)x + 4u + 3}$, such that 
    \[
        \mathcal{L}(P_\infty+P_{1,u+2}+P_{1,4u+3}+P_{3u+3,2}+P_{3u+3,3}+P_{2u+3,2u+4})=\langle\,e_1,e_2,e_3,e_4,e_5\,\rangle.
    \]
    In addition, there are five other sets of evaluation points 
    \[
        \{P_{2u,u+2},P_{2u,4u+3},P_{3u+4,1},P_{3u+4,4},P_{4,2u+4},P_{4,3u+1}\},
    \]
    \[
        \{P_{4u+4,u+2},P_{4u+4,4u+3},P_{2u+2,1},P_{2u+2,4},P_{u,1},P_{u,4}\},
    \]
    \[
        \{P_{u+4,2},P_{u+4,3},P_{2u+1,2},P_{2u+1,3},P_{4u+2,2u+4},P_{4u+2,3u+1}\},
    \]
    \[
        \{P_{3u+2,u+2},P_{3u+2,4u+3},P_{2,2},P_{2,3},P_{u+1,2u+4},P_{u+1,3u+1}\},
    \]
    \[
        \text{and } \{P_{u+3,u+2},P_{u+3,4u+3},P_{3,1},P_{3,4},P_{3u,2u+4},P_{3u,3u+1}\}.
    \]

    By Proposition \ref{con_1} and Sagemath, we can construct the following optimal $5^2$-ary locally repairable codes with locality $5$ in Table \ref{opt_LRC_y2=x5+x}.
    \begin{table}[ht]
	\centering
	\begin{tabular}{c}
        \hline
		Our optimal locally repairable codes \\

        $[n=6\ell,\,k=5t-4,\,d=n-6(t-1)]_{5^2}$\\
        \hline
		$[12,\,6,\,6]_{5^2}$ \\

		$[18,\,11,\,6]_{5^2}$ \\

        $[24,\,16,\,6]_{5^2}$ \\

        $[30,\,21,\,6]_{5^2}$ \\

        $[36,\,26,\,6]_{5^2}$ \\
        \hline
	\end{tabular}
	\caption{Optimal locally repairable codes with locality $5$ over $\mathbb{F}_{5^2}$}
	\label{opt_LRC_y2=x5+x}
    \end{table}
\end{example}

\begin{remark}
    Note that, the lengths of the optimal $[16,\,10,\,4]_{3^2}$ and $[36,\,26,\,6]_{5^2}$ locally repairable codes in Example \ref{r_odd_optimal} are longer than $q+2\sqrt{q}$, that is to say, the length of our optimal locally repairable code could exceed the bound in \cite{Optimal_Locally_Repairable_Codes_Via_Elliptic_Curves}. 
\end{remark}

\subsection{Locality r even}

Besides the subgroups $\mathcal{G}$ of ${\rm{Aut}}(\mathfrak{C}/\mathbb{F}_q)$ previously discussed where $|\mathcal{G}|$ is an even integer, there exist subgroups of odd order that can be employed to construct locally repairable codes with even locality.

\begin{theorem}\label{r_even}
    Let $q=p^{2a}$ be a power of an odd prime $p$ and an odd number $a>0$. If $p\equiv -1 \pmod{5}$ or $p=5$, then there exists an either optimal or almost optimal q-ary
    \[
        [n=5\ell,\,k=4t+1,\,d\ge n-5t-1]_q
    \]
    locally repairable code with locality $4$ for any $0\le t<\ell\le\big\lfloor\frac{q+4\sqrt{q}-11}{5}\big\rfloor$.
\end{theorem}
\begin{proof}
    (i) If $p\equiv -1 \pmod{5}$, then by the proof of Theorem \ref{r_odd}(2), we know that the hyperelliptic curve $\mathfrak{C}:y^2=x^5+1$ is maximal over $\mathbb{F}_q$. We consider the subgroup $\mathcal{G}=\langle V\rangle$ where
    \[
        V=\begin{pmatrix}
        \zeta^2 & 0 \\
        0 & \zeta
        \end{pmatrix}\,({\rm ord}(\zeta)=5)
    \]
    and $|\mathcal{G}|=5$. Let $F$ be the function field $\mathbb{F}_q(\mathfrak{C})$. By (\ref{trans_eq}), we can conclude that $\sigma(P_\infty)=P_\infty$ for any $\sigma\in\mathcal{G}$, which means that $P_\infty\cap F^{\mathcal{G}}$ is totally ramified in $F/F^{\mathcal{G}}$. Hence, by the Hurwitz Genus Formula (\ref{Hurwitz_genus_formula}), we have
    \[
        2=2g(F)-2\ge [F:F^{\mathcal{G}}](2g(F^{\mathcal{G}})-2)+[F:F^{\mathcal{G}}]-1=5(2g(F^{\mathcal{G}})-2)+4.
    \]
    It follows that $g(F^{\mathcal{G}})=0$ i.e. $F^{\mathcal{G}}$ is a rational function field. For this reason, we can choose $z\in F^{\mathcal{G}}$ such that $(z)_{\infty}^{F^{\mathcal{G}}}=P_\infty\cap F^{\mathcal{G}}$. Then we have $(z)_{\infty}^{F}=5P_\infty$. In addition, we know that $(x)_{\infty}^{F}=2P_\infty$. 
    
    Now consider the set of functions
    \[
        V^\prime=\Big\{ f_0(z)+f_1(z)x+f_2(z)x^2+f_3(z)x^3:f_i(z)\in\mathbb{F}_q[z];\,\deg(f_0)\le t;\,\deg(f_j)\le t-1,\,j=1,2,3 \Big\}.
    \]
    It can be verified that $V$ is a subspace of $\mathcal{L}((5t+1)P_\infty)$. Then we prove that $\dim_{\mathbb{F}_q} V=4t+1$. Suppose that 
    \[
        f_0(z)+f_1(z)x+f_2(z)x^2+f_3(z)x^3=0,
    \]
    where $f_i(z)\,(0\le i\le 3)$ are not all zero polynomials. Since 
    \[
        v_{P_\infty}(f_i(z)x^i)=\deg(f_i(z))\cdot v_{P_\infty}(z)+i\cdot v_{P_\infty}(x)\equiv 3i \bmod{5},
    \]
    for $f_i(z)\neq 0$ and $i=0,1,2,3$. Therefore $v_{P_\infty}(f_i(z)x^i)\neq v_{P_\infty}(f_j(z)x^j)$ whenever $i\neq j$ and $f_i(z)\neq 0$. The Strict Triangle Inequality yields
    \[
        v_{P_\infty}\Big(\sum_{i=0}^3f_i(z)x^i\Big)=\mathop{\min}_{0\le i\le 3} \{ \,v_{P_\infty}(f_i(z)x^i):f_i(z)\neq 0 \,\}<\infty,
    \]
    which is a contradiction. It implies that $\dim_{\mathbb{F}_q} V=4t+1$. 
    
    Let $\mathcal{P}$ be defined as (\ref{ev_P}). Define the algebraic geometry code
    \[
        C_{\mathcal{L}}(\mathcal{P},V^\prime)=\{(f(P))_{P\in\mathcal{P}}:f\in V^\prime\}.
    \]
    Then the code $C_{\mathcal{L}}(\mathcal{P},V^\prime)$ is a $q$-ary 
    \[
        [n=5\ell,\,k=4t+1,\,d\ge n-5t-1]_q
    \]
    linear code. Note that every $4\times 4$ submatrix of the matrix 
    \[
        \begin{pmatrix}\,
        1 & x(P_{i,1}) & x^2(P_{i,1}) & x^3(P_{i,1}) \\
        1 & x(P_{i,2}) & x^2(P_{i,2}) & x^3(P_{i,2}) \\
        1 & x(P_{i,3}) & x^2(P_{i,3}) & x^3(P_{i,3}) \\
        1 & x(P_{i,4}) & x^2(P_{i,4}) & x^3(P_{i,4}) \\
        1 & x(P_{i,5}) & x^2(P_{i,5}) & x^3(P_{i,5})
        \,\end{pmatrix}
    \]
    is a Vandermonde matrix, therefore it is invertible. Hence, by Lemma \ref{local}$, C_{\mathcal{L}}(\mathcal{P},V^\prime)$ is also a locally repairable code with locality 4. The bound (\ref{min_d_bound}) shows that  
    \begin{align*}
        d&\le n-k-\bigg\lceil\frac{k}{r}\bigg\rceil+2\\
        &=n-4t-1-\bigg\lceil\frac{4t+1}{4}\bigg\rceil+2\\
        &=n-5t,
    \end{align*}
    which implies that $C_{\mathcal{L}}(\mathcal{P},V^\prime)$ is an either optimal or almost optimal locally repairable code with locality 4. The upper bound of $\ell$ is easy to obtain from the proof of Proposition \ref{con_1}. 

    (ii) If $p=5$, then by the proof of Theorem \ref{r_odd}(1), we know that $\mathfrak{C}:y^2=x^5+x$ is maximal over $\mathbb{F}_q$. We consider the subgroup $\mathcal{G}=\langle H=(UV)^2 \rangle$ where
    \[
        U=\begin{pmatrix}
        0 & -(-1)^{-1/4}\cdot2 \\
        -(-1)^{1/4}\cdot2 & 0
        \end{pmatrix}\quad {\rm and} \quad V=\begin{pmatrix}
        0 & -(-1)^{-1/4}\cdot2 \\
        -(-1)^{1/4}\cdot2 & 1
        \end{pmatrix}.
    \]
    By calculation, we can get
    \[
        H=\begin{pmatrix}
        1 & -(-1)^{-1/4} \\
        0 & 1
        \end{pmatrix}
    \]
    and $|H|=5$. Let $F$ be the function field $\mathbb{F}_q(\mathfrak{C})$. By (\ref{trans_eq}), we can also conclude that $\sigma(P_\infty)=P_\infty$ for any $\sigma\in\mathcal{G}$. The rest of the proof is completely similar to the Part (i) above. 
\end{proof}

\begin{remark}\label{r_even_re}
    {\rm (1)} For the case of locality $r=2$, we can consider the hyperelliptic curve $\mathfrak{C}:y^2=x^6+x^3+\tau$ over $\mathbb{F}_q$ in Lemma \ref{auto_non}(ii). Let $q$ be a power of an odd prime other than 3. Suppose that $3\mid q-1$, then we have $\mathcal{G}=\langle H=U^2 \rangle\subseteq{\rm{Aut}}(\mathfrak{C}/\mathbb{F}_q)$ by Lemma \ref{auto_non}(ii) where
    \[
        U=\begin{pmatrix}
        -\alpha^2 & 0 \\
        0 & -\alpha
        \end{pmatrix}\,({\rm ord}(\alpha)=3).
    \]
    By calculation, we can get
    \[
        H=\begin{pmatrix}
        \alpha & 0 \\
        0 & \alpha^2
        \end{pmatrix}
    \]
    and $|H|=3$. Let $F$ be the function field $\mathbb{F}_q(\mathfrak{C})$. By (\ref{trans_eq}), we conclude that $\sigma(D_\infty)=D_\infty$ for any $\sigma\in\mathcal{G}$. The Hurwitz Genus Formula (\ref{Hurwitz_genus_formula}) yields
    \[
        2=2g(F)-2\ge [F:F^{\mathcal{G}}](2g(F^{\mathcal{G}})-2)+([F:F^{\mathcal{G}}]-1)\deg D_\infty=3(2g(F^{\mathcal{G}})-2)+4.
    \]
    It follows that $F^{\mathcal{G}}$ is a rational function field. Moreover, for any $\sigma\in\mathcal{G}$, we also have 
    \[
        \sigma(P_{0,\tau^{1/2}})=P_{0,\tau^{1/2}}\quad  {\it and} \quad \sigma(\bar{P}_{0,\tau^{1/2}})=\bar{P}_{0,\tau^{1/2}}.
    \]
    By the Hurwitz Genus Formula, the places which are ramified in $F/F^{\mathcal{G}}$ are exactly 
    \[
        P_{\infty^+}\cap F^{\mathcal{G}},\;P_{\infty^-}\cap F^{\mathcal{G}},\;P_{0,\tau^{1/2}}\cap F^{\mathcal{G}},\;\bar{P}_{0,\tau^{1/2}}\cap F^{\mathcal{G}}. 
    \]
    
    Choose $z\in F^{\mathcal{G}}$ such that $(z)_{\infty}^{F}=3P_{\infty^+}$ and $x\in F$ such that $(x)_{\infty}^{F}=P_{\infty^+}+P_{\infty^-}$. Consider the $\mathbb{F}_q$-space 
    \[
        V^\prime=\Big\{ f_0(z)+f_1(z)x:f_i(z)\in\mathbb{F}_q[z];\,\deg(f_0)\le t;\,\deg(f_1)\le t-1 \Big\}.
    \]
    Similar to the proof of Theorem \ref{r_even}, we can obtain $C_{\mathcal{L}}(\mathcal{P},V^\prime)$ is an either optimal or almost optimal q-ary
    \[
        [n=3\ell,\,k=2t+1,\,d\ge n-3t-1]_q
    \]
    locally repairable code with locality $2$ for any $1\le t<\ell\le\big\lfloor\frac{N}{3}\big\rfloor-1$ where $N=\#\mathbb{P}_F^1$. 

    {\rm (2)} In the proof of Theorem \ref{r_even}, if we take $\deg(f_1)\le t$ and keep the rest unchanged, then $C_{\mathcal{L}}(\mathcal{P},V^\prime)$ will be an either optimal or almost optimal q-ary
    \[
        [n=5\ell,\,k=4t+2,\,d\ge n-5t-2]_q
    \]
    locally repairable code with locality $4$ for any $0\le t<\ell\le\big\lfloor\frac{q+4\sqrt{q}-11}{5}\big\rfloor$. 
\end{remark}

\begin{example}\label{r_even_alm_opt}
    {\rm (1)} Let $q=5^2$ and the hyperelliptic curve $\mathfrak{C}:y^2=x^6+x^3+2$. Choose $z=y+4x^3$ in Remark \ref{r_even_re}(1). Then we can obtain an either optimal or almost optimal $5^2$-ary
    \[
        [n=3\ell,\,k=2t+1,\,d\ge n-3t-1]_{5^2}
    \]
    locally repairable code with locality $2$ for any $0\le t<\ell\le12$. By Sagemath, we can obtain some long optimal $5^2$-ary locally repairable codes with locality $2$ in Table \ref{LRC_y2=x6+x3+2_opt}.
    \begin{table}[ht]
	\centering
	\begin{tabular}{c}
        \hline
		Our optimal locally repairable codes \\

        $[n=3\ell,\,k=2t+1,\,d=n-3t]_{5^2}$ \\
        \hline
        $[30,\,17,\,6]_{5^2}$ \\

        $[30,\,19,\,3]_{5^2}$ \\
        
        $[33,\,19,\,6]_{5^2}$ \\

        $[33,\,21,\,3]_{5^2}$ \\

        $[36,\,21,\,6]_{5^2}$ \\

        $[36,\,23,\,3]_{5^2}$ \\
        \hline
	\end{tabular}
	\caption{Optimal locally repairable codes with locality $2$ over $\mathbb{F}_{5^2}$}
	\label{LRC_y2=x6+x3+2_opt}
    \end{table}

    {\rm (2)} Let $q=5^2$ and the hyperelliptic curve $\mathfrak{C}:y^2=x^5+x$. Choose $z=y$ in Theorem \ref{r_even}. Then we can obtain some either optimal or almost optimal $5^2$-ary locally repairable codes with locality $4$ in Table \ref{LRC_y2=x5+x_3}.
    \begin{table}[ht]
	\centering
	\begin{tabular}{cc}
        \hline
		Our locally repairable codes & The range of $\ell$ \\
        \hline
		$[n=5\ell,\,k=4t+1,\,d\ge n-5t-1]_{5^2}$ & $0\le t<\ell\le9$ \\

		$[n=5\ell,\,k=4t+2,\,d\ge n-5t-2]_{5^2}$ & $0\le t<\ell\le9$ \\
        \hline
	\end{tabular}
	\caption{Locally repairable codes with locality $4$ over $\mathbb{F}_{5^2}$}
	\label{LRC_y2=x5+x_3}
    \end{table}

    By Sagemath, we can obtain some optimal $5^2$-ary locally repairable codes with locality $4$ in Table and \ref{LRC_y2=x5+x_3_opt_2}.

    \begin{table}[ht]
	\centering
	\begin{tabular}{cc}
        \hline
        \multicolumn{2}{c}{Our optimal locally repairable codes}          \\		

        $[n=5\ell,\,k=4t+1,\,d=n-5t]_{5^2}$ & $[n=5\ell,\,k=4t+2,\,d=n-5t-1]_{5^2}$ \\
        \hline
        $[10,\,5,\,5]_{5^2}$ & $[25,\,2,\,24]_{5^2}$ \\

        $[15,\,5,\,10]_{5^2}$ & $[25,\,18,\,4]_{5^2}$ \\

        $[15,\,9,\,5]_{5^2}$ & $[30,\,22,\,4]_{5^2}$ \\

        $[20,\,9,\,10]_{5^2}$ & $[35,\,26,\,4]_{5^2}$ \\

        $[20,\,13,\,5]_{5^2}$ & $[40,\,30,\,4]_{5^2}$ \\

        $[25,\,17,\,5]_{5^2}$ & $[45,\,34,\,4]_{5^2}$ \\
        \hline
	\end{tabular}
	\caption{Optimal locally repairable codes with locality $4$ over $\mathbb{F}_{5^2}$}
	\label{LRC_y2=x5+x_3_opt_2}
    \end{table}

\end{example}

\begin{remark}
    {\rm (1)} Note that the lengths of the optimal $[36,\,21,\,6]_{5^2}$, $[36,\,23,\,3]_{5^2}$, $[40,\,30,\,4]_{5^2}$ and $[45,\,34,\,4]_{5^2}$ locally repairable codes in Example \ref{r_even_alm_opt}(2) are longer than $q+2\sqrt{q}$, and even the length of the last one reaches $q+4\sqrt{q}$. 

    {\rm (2)} As we can see, the optimal $[36,\,21,\,6]_{5^2}$ and $[36,\,23,\,3]_{5^2}$ locally repairable codes can not be obtained from Theorem 1 in \cite{Optimal_Locally_Repairable_Codes_Via_Elliptic_Curves}. To some extent, the length of our optimal locally repairable code could be longer than the code constructed in \cite{Optimal_Locally_Repairable_Codes_Via_Elliptic_Curves}. 
\end{remark}

\section{Conclusion}\label{conclusion}

In this paper, we introduce a method for constructing locally repairable codes through the use of automorphism groups associated with genus $2$ hyperelliptic curves. This approach leads to the development of several new families of locally repairable codes that are either optimal or almost optimal. The maximum code length $n$ attainable through this construction can approach $q+4\sqrt{q}$, and the locality $r$ satisfies $r+1\mid |{\rm{Aut}}(\mathfrak{C}/\mathbb{F}_q)|$ (refer to Theorems \ref{r_odd} and \ref{r_even}). Moreover, we derive certain optimal locally repairable codes with lengths surpassing the limit of $q+2\sqrt{q}$ (as illustrated in Examples \ref{r_odd_optimal} and \ref{r_even_alm_opt}). Future research will need to focus on developing criteria to determine the conditions under which the codes we have constructed are indeed optimal.

\section*{Acknowledgement}

This work is supported by the National Natural Science Foundation of China (No. 12441107), Guangdong Major Project of Basic and Applied Basic Research
(No. 2019B030302008) and Guangdong Provincial Key Laboratory of Information Security Technology (No.2023B1212060026).

\bibliographystyle{ieeetr}
\bibliography{reference}

\begin{thebibliography}{10}

\bibitem{On_the_Locality_of_Codeword_Symbols}
P.~Gopalan, C.~Huang, H.~Simitci, and S.~Yekhanin, ``On the locality of
  codeword symbols,'' {\em IEEE Transactions on Information Theory}, vol.~58,
  no.~11, pp.~6925--6934, 2012.

\bibitem{Locally_Recoverable_Codes_from_Algebraic_Curves_and_Surfaces}
A.~Barg, K.~Haymaker, E.~W. Howe, G.~L. Matthews, and A.~V{\'a}rilly-Alvarado,
  ``Locally recoverable codes from algebraic curves and surfaces,'' in {\em
  Algebraic Geometry for Coding Theory and Cryptography} (E.~W. Howe, K.~E.
  Lauter, and J.~L. Walker, eds.), (Cham), pp.~95--127, Springer International
  Publishing, 2017.

\bibitem{A_Family_of_Optimal_Locally_Recoverable_Codes}
I.~Tamo and A.~Barg, ``A family of optimal locally recoverable codes,'' {\em
  IEEE Transactions on Information Theory}, vol.~60, no.~8, pp.~4661--4676,
  2014.

\bibitem{Locally_Recoverable_Codes_on_Algebraic_Curves}
A.~Barg, I.~Tamo, and S.~Vlăduţ, ``Locally recoverable codes on algebraic
  curves,'' {\em IEEE Transactions on Information Theory}, vol.~63, no.~8,
  pp.~4928--4939, 2017.

\bibitem{Optimal_Locally_Repairable_Codes_Via_Elliptic_Curves}
X.~Li, L.~Ma, and C.~Xing, ``Optimal locally repairable codes via elliptic
  curves,'' {\em IEEE Transactions on Information Theory}, vol.~65, no.~1,
  pp.~108--117, 2019.

\bibitem{Locally_Recoverable_J-affine_variety_codes}
C.~Galindo, F.~Hernando, and C.~Munuera, ``Locally recoverable {J}-affine
  variety codes,'' {\em Finite Fields and Their Applications}, vol.~64,
  p.~101661, 2020.

\bibitem{Locally_Recoverable_codes_with_availability_<i>t</i>≥2_from_fiber_products_of_curves}
K.~Haymaker, B.~Malmskog, and G.~L. Matthews, ``Locally recoverable codes with
  availability t$\ge$ 2 from fiber products of curves,'' {\em Advances in
  Mathematics of Communications}, vol.~12, no.~2, pp.~317--336, 2018.

\bibitem{Locally_Recoverable_codes_from_rational_maps}
C.~Munuera and W.~Tenório, ``Locally recoverable codes from rational maps,''
  {\em Finite Fields and Their Applications}, vol.~54, pp.~80--100, 2018.

\bibitem{Locally_Recoverable_codes_from_algebraic_curves_with_separated_variables}
C.~Munuera, W.~Tenório, and F.~Torres, ``Locally recoverable codes from
  algebraic curves with separated variables,'' {\em Advances in Mathematics of
  Communications}, vol.~14, no.~2, pp.~265--278, 2020.

\bibitem{Construction_of_Optimal_Locally_Repairable_Codes_via_Automorphism_Groups_of_Rational_Function_Fields}
L.~Jin, L.~Ma, and C.~Xing, ``Construction of optimal locally repairable codes
  via automorphism groups of rational function fields,'' {\em IEEE Transactions
  on Information Theory}, vol.~66, no.~1, pp.~210--221, 2020.

\bibitem{Constructions_of_optimal_and_almost_optimal_locally_repairable_codes}
T.~Ernvall, T.~Westerbäck, and C.~Hollanti, ``Constructions of optimal and
  almost optimal locally repairable codes,'' in {\em 2014 4th International
  Conference on Wireless Communications, Vehicular Technology, Information
  Theory and Aerospace, Electronic Systems (VITAE)}, pp.~1--5, 2014.

\bibitem{Local_codes_with_addition_based_repair}
H.~M. Kiah, S.~H. Dau, W.~Song, and C.~Yuen, ``Local codes with addition based
  repair,'' in {\em 2015 IEEE Information Theory Workshop - Fall (ITW)},
  pp.~59--63, 2015.

\bibitem{Prolegomena_to_a_Middlebrow_Arithmetic_of_Curves_of_Genus_2}
J.~W.~S. Cassels and E.~V. Flynn, ``Prolegomena to a middlebrow arithmetic of
  curves of genus 2,'' London Mathematical Society Lecture Note Series,
  Cambridge University Press, 1996.

\bibitem{Construction_de_courbes_de_genre_2_`a_partir_de_leurs_modules}
J.-F. Mestre, ``Construction de courbes de genre 2 {\`a} partir de leurs
  modules,'' in {\em Effective Methods in Algebraic Geometry} (T.~Mora and
  C.~Traverso, eds.), (Boston, MA), pp.~313--334, Birkh{\"a}user Boston, 1991.

\bibitem{Computational_aspects_of_curves_of_genus_at_least_2}
B.~Poonen, ``Computational aspects of curves of genus at least 2,'' in {\em
  Algorithmic Number Theory} (H.~Cohen, ed.), (Berlin, Heidelberg),
  pp.~283--306, Springer Berlin Heidelberg, 1996.

\bibitem{Algebraic_Function_Fields_and_Codes}
H.~Stichtenoth, ``Algebraic function fields and codes,'' in {\em Graduate Texts
  in Mathematics 254}, Springer Berlin, Heidelberg, 2009.

\bibitem{A_note_on_certain_maximal_hyperelliptic_curves}
S.~Tafazolian, ``A note on certain maximal hyperelliptic curves,'' {\em Finite
  Fields and Their Applications}, vol.~18, no.~5, pp.~1013--1016, 2012.

\bibitem{Efficient_Hyperelliptic_Arithmetic_Using_Balanced_Representation_for_Divisors}
S.~D. Galbraith, M.~Harrison, and D.~J. Mireles~Morales, ``Efficient
  hyperelliptic arithmetic using balanced representation for divisors,'' in
  {\em Algorithmic Number Theory} (A.~J. van~der Poorten and A.~Stein, eds.),
  (Berlin, Heidelberg), pp.~342--356, Springer Berlin Heidelberg, 2008.

\bibitem{Algebraic-Geometric_Codes}
M.~A. Tsfasman and S.~G. Vlăduţ, ``Algebraic-geometric codes,'' in {\em
  Mathematics and its Applications}, Springer Dordrecht, 1991.

\bibitem{On_Binary_Sextics_with_Linear_Transformations_into_Themselves}
O.~Bolza, ``On binary sextics with linear transformations into themselves,''
  {\em American Journal of Mathematics}, vol.~10, no.~1, pp.~47--70, 1887.

\bibitem{On_curves_of_genus_2_with_Jacobian_of_GL_2-type}
G.~Cardona, J.~Gonz{\'a}lez, J.-C. Lario, and A.~Rio, ``On curves of genus 2
  with jacobian of ${GL}_2$-type,'' {\em manuscripta mathematica}, vol.~98,
  pp.~37--54, 1999.

\bibitem{Curves_of_genus_2_with_group_of_automorphisms_isomorphic_to_D_8_or_D12}
G.~Cardona and J.~Quer, ``Curves of genus 2 with group of automorphisms
  isomorphic to ${D}_8$ or ${D}_{12}$,'' {\em Transactions of the American
  Mathematical Society}, vol.~359, pp.~2831--2849, 2002.

\bibitem{On_the_number_of_curves_of_genus_2_over_a_finite_field}
G.~Cardona, ``On the number of curves of genus 2 over a finite field,'' {\em
  Finite Fields and Their Applications}, vol.~9, no.~4, pp.~505--526, 2003.

\end{thebibliography}

\end{document}